\documentclass[conference]{IEEEtran}
\usepackage{times}
\newcommand{\HRule}
{\noindent\rule{\linewidth}{0.1mm}\newline}
\usepackage{graphicx}
\usepackage{subcaption}
\usepackage[numbers,sort&compress]{natbib}
\let\labelindent\relax
\usepackage{enumitem}
\usepackage[bookmarks=true]{hyperref}
\usepackage{amsmath,amsthm,amsfonts}
\newtheorem{lemma}{Lemma}
\newtheorem{assumption}{Assumption}
\newtheorem{definition}{Definition}
\newtheorem{proposition}{Proposition}
\newtheorem{remark}{Remark}
\newtheorem{theorem}{Theorem}
\newtheorem{corollary}{Corollary}

\usepackage{color}
\usepackage{mdframed}
\usepackage{amssymb}
\usepackage{amsmath}
\usepackage{multicol}
\usepackage{xcolor}
\usepackage{flushend}

\newcommand{\mysquare}[1][black]{\textcolor{#1}{\ensuremath\blacksquare}}

\definecolor{mplblue}{HTML}{1f77b4}
\definecolor{mplgreen}{HTML}{2ca02c}

\DeclareMathOperator*{\argmin}{arg\,min}

\begin{document}

\title{Computation-Aware Learning for Stable Control with Gaussian Process}



\author{\authorblockN{Wenhan Cao$^{1,2}$ 
Alexandre Capone$^{3,4}$
Rishabh Yadav$^{1}$  
Sandra Hirche$^{4}$ Wei Pan$^{1}$}
\authorblockA{$^{1}$University of Manchester~~~$^{2}$Tsinghua University~~~$^{3}$Carnegie Mellon University~~~$^{4}$Technical University of Munich}}

\maketitle

\begin{abstract}
In Gaussian Process (GP) dynamical model learning for robot control, particularly for systems constrained by computational resources like small quadrotors equipped with low-end processors, analyzing stability and designing a stable controller present significant challenges. This paper distinguishes between two types of uncertainty within the posteriors of GP dynamical models: the well-documented mathematical uncertainty stemming from limited data and computational uncertainty arising from constrained computational capabilities, which has been largely overlooked in prior research. Our work demonstrates that computational uncertainty, quantified through a probabilistic approximation of the inverse covariance matrix in GP dynamical models, is essential for stable control under computational constraints. We show that incorporating computational uncertainty can prevent overestimating the region of attraction, a safe subset of the state space with asymptotic stability, thus improving system safety. Building on these insights, we propose an innovative controller design methodology that integrates computational uncertainty within a second-order cone programming framework. Simulations of canonical stable control tasks and experiments of quadrotor tracking exhibit the effectiveness of our method under computational constraints.
\end{abstract}

\IEEEpeerreviewmaketitle
\footnotetext{\fontsize{8}{8} Website: {\href{https://sites.google.com/view/computation-gp}{\texttt{https://sites.google.com/view/computation-gp}}}
\\ 
Correspondence to {\href{wei.pan@manchester.ac.uk}{\texttt{{wei.pan}@manchester.ac.uk}}}}

\section{Introduction}
Dynamical model learning using Gaussian processes (GP) is popular in the field of robotic control \cite{nguyen2008local,deisenroth2013gaussian,jang2020multi}. The ability to update the model online is particularly desirable as it enables robotic systems to adapt to unpredictable and changing environments \cite{haddadi2008online,nguyen2010real,fang2019vision,wilcox2020solar}. Using dynamically updated models makes it feasible to design controllers with stability guarantees \cite{wen1990unified, santibanez1997strict, khansari2014learning, giesl2015review, nguyen2015optimal}. This adaptability is crucial in aerial robotics, where the capacity to quickly adjust to new conditions is key to maintaining both stability and optimal performance. Existing approaches assume ample computational resources when learning dynamical models online \cite{berkenkamp2016safe,fisac2018general,wang2018safe,umlauft2018uncertainty,lederer2019uniform,beckers2019stable,khan2021safety,castaneda2021gaussian}. However, low-end computational hardware, such as those of tiny robots \cite{neuman2022tiny}, often faces limitations, and even high-end systems can suffer performance degradation due to insufficient resource allocation or high temperatures \cite{zhou2010performance}. Consequently, computational errors are inevitable, whether due to the early termination of the learning process or the processing of only a partial batch of data.

Computational errors in  dynamical model learning can deteriorate control performance. An illustrative example can be found in our experiment (see Section~\ref{exp:quad}), where we consider a quadrotor tracking control task. This task utilizes the conjugate gradient (CG) method, an iterative approach, to compute the GP posterior dynamical model. We then synthesize a controller with constraints constructed by the control Lyapunov function (CLF) \cite{sontag1989universal,castaneda2021gaussian,khan2021safety}. The results demonstrate that early termination of the CG algorithm can lead to increased tracking error (see Fig.~\ref{fig.tracking error}). Intuitively, under a fixed operation time budget, low-end computational hardware allows fewer iterations, resulting in larger model errors, and consequently, greater tracking errors than its high-end counterpart.

\begin{figure}
    \centering
    \includegraphics[width=1\linewidth]{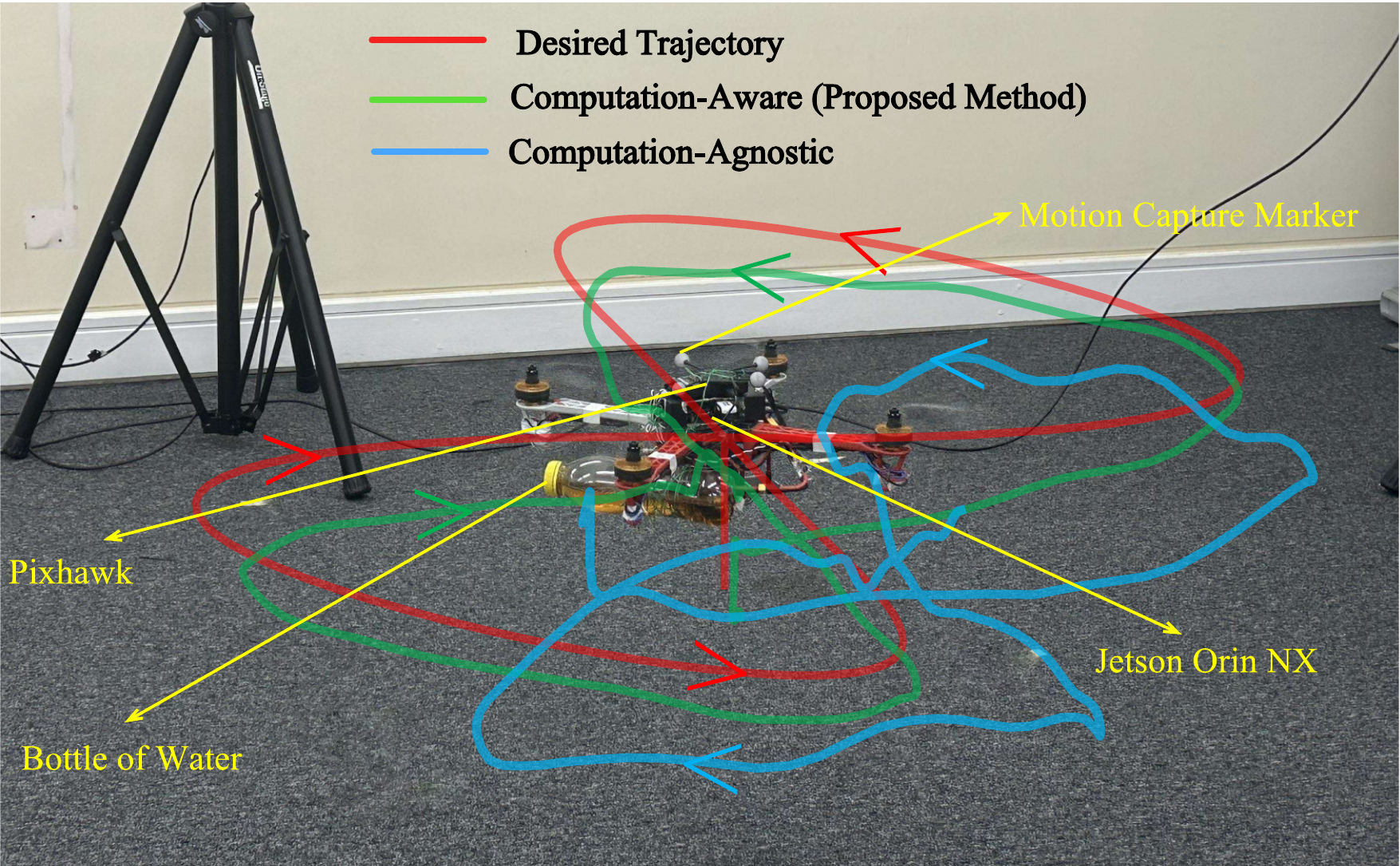}
    \caption{Quadrotor tracking control under computational constraints. The green line denotes the proposed computation-aware learning control method in this paper, showing improved tracking accuracy due to adaptive, conservative stability constraints compared to the agnostic one (blue line).}
    \vspace{-0.35cm}
    \label{fig.intro}
\end{figure}

\begin{figure*}[!t]
    \centering
    \includegraphics[width=1.0\linewidth]{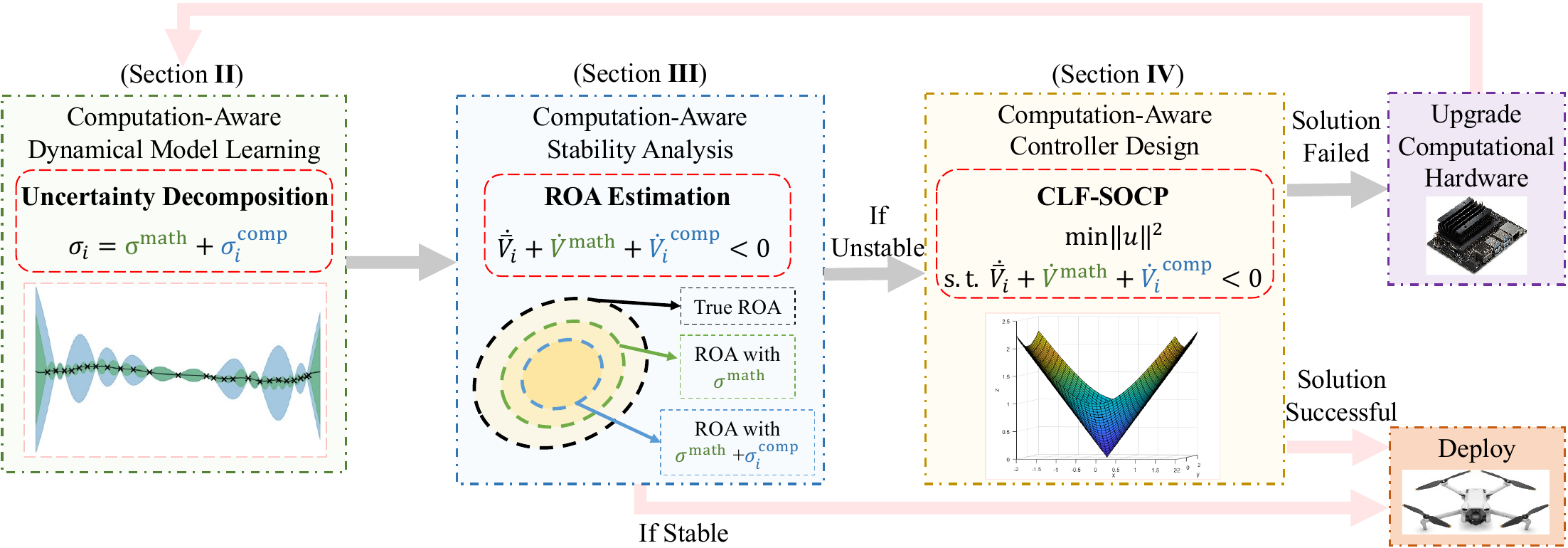}
    \caption{Workflow of computation-aware learning for stable control. This figure illustrates the three main stages of developing a control system with computational constraints. Initially, in Section \ref{sec.model learning} (computation-aware dynamical model learning), the model's uncertainty is broken down into mathematical and computational parts. Following that, in Section \ref{sec.stability analysis} (computation-aware stability analysis), the system's stability is analyzed by estimating its Region of Attraction (ROA) under various uncertainties. In Section \ref{sec.controller design} (computation-aware controller design), a controller is designed to ensure system stability with minimal control effort, using the control lyapunov function-second order cone programming (CLF-SOCP) method. The flowchart concludes with a decision-making process that evaluates whether the system's computational capacity needs upgrading before deployment.}
    \vspace{-0.5cm}
    \label{fig.outline}
\end{figure*}

Given the inevitability of computational errors during dynamical model learning, understanding their impact on system stability and incorporating them into controller design is crucial. To address this, this paper proposes a computation-aware learning framework for stable control, which includes computation-aware dynamical model learning, computation-aware stability analysis, and computation-aware controller design, as described in Fig.~\ref{fig.outline}. The contributions of this paper can be summarized as follows:

\begin{itemize}
    \item We demonstrate that the posterior of a GP dynamical model for a control-affine system comprises two kinds of uncertainty: mathematical uncertainty and computational uncertainty, arising due to limited data and constrained computation, respectively. This is achieved by exploiting the iterative approximation of the inverse of the covariance matrix in a probabilistic manner.
   
    \item We quantify the impact of computation on the stability of learning-based systems by examining the derivative of the Lyapunov function along the trajectories of the dynamical system. Based on this stability condition, the region of attraction (ROA) of the real system is estimated by discretizing the state space with Lipschitz conditions.
    
    \item We formulate a second-order cone programming (SOCP) approach to synthesize a minimum-norm stabilizing controller, taking into account both the mathematical and computational uncertainties of the learned dynamical model by leveraging the control Lyapunov function. Additionally, we demonstrate that an explicit control formula can be constructed, extending Sontag's universal formula to include both forms of uncertainty.
\end{itemize}

The remainder of this paper is organized as follows: Section \ref{sec.model learning} introduces computation-aware GP dynamical model learning. Section \ref{sec.stability analysis} presents the computation-aware stability analysis. Section \ref{sec.controller design} describes the computation-aware controller design. The conclusions and discussions are provided in Section~\ref{sec.conclusion}.


\section{Computation-Aware Model Learning}\label{sec.model learning}
In this section, we address the following question:
\begin{mdframed}
\textbf{Q1}: How can the computational error in dynamical model learning be quantified? 
\end{mdframed}
The key idea is to regard the computation of the inverse of the covariance
matrix in traditional GP dynamical model learning as a probabilistic inference problem. Consider a continuous-time control affine system 
\begin{equation}\label{eq.system model}
\dot{x} = f(x) + g(x)u,  
\end{equation}
with system state \( x \in \mathcal{X} \subseteq \mathbb{R}^n \) and control input \( u \in \mathcal{U} \subseteq \mathbb{R}^m\). In \eqref{eq.system model}, the symbols $f: \mathbb{R}^n \to \mathbb{R}^n$ and $g: \mathbb{R}^n \to \mathbb{R}^{n \times m}$ denote unknown functions that represent the drift dynamics and
the input matrix, respectively. 
We assume that the unknown functions $f$ and $g$ have bounded reproducing kernel Hilbert space (RKHS) norm \cite{berlinet2011reproducing}, which is a typical assumption when using GP in robotics \cite{berkenkamp2016safe,beckers2019stable}:
\begin{assumption}[Bounded RKHS norm]\label{assump.bounded RKHS}
We assume that $f \in \mathcal{H}_f$ and $g \in \mathcal{H}_g$, where $\mathcal{H}_f$ and $\mathcal{H}_g$ represent the RKHSs induced by continuously differentiable and bounded kernels $k^f(x, x^{\prime})$ and $k^g(x, x^{\prime})$, i.e., $f$ and $g$ have a bounded RKHS norm with respect to continuously differentiable bounded kernels $k^f(x, x^{\prime})$ and $k^g(x, x^{\prime})$. Furthermore, we know corresponding upper bounds $B_f$ and $B_g$, i.e.,  $\left\|f \right\|_{\mathcal{H}_f} \leq B_f$ and $\left\|g \right\|_{\mathcal{H}_g} \leq B_g$.
\end{assumption} 
Initially, we assume that we have access to measurements \( \dot{x} = f(x) + g(x)u\), which is a common assumption for GP dynamical model learning \cite{berkenkamp2016safe,umlauft2019feedback,lederer2020training}.
In practice, the derivative of the state $\dot{x}$ could be difficult to acquire and often substituted with a discrete-time approximation \cite{berkenkamp2016safe}. 
One notable challenge is to design the GP prior $f(x) + g(x)u \sim \mathcal{GP} (\mu(x, u), k(x,u,x^{\prime},u^{\prime}) )$ to differentiate the effect of the drift dynamics $f$ and the input matrix $g$ for the control-affine system \eqref{eq.system model}. 
Following the structure of a control affine system, as suggested in \cite{umlauft2019feedback,castaneda2021gaussian}, we choose the GP prior mean to be the nominal model, denoted as \(\mu(x, u) = \hat{f}(x) + \hat{g}(x)u\). Here, \(\hat{f}(x)\) and \(\hat{g}(x)\) represent the nominal parts of the drift dynamics and the control matrix, respectively, which can be modeled with relative ease. Furthermore, we select the GP kernel \(k(x,u, x^{\prime},u^{\prime})\) as the composite GP kernel, defined as 
\begin{equation}\label{eq.composite kernel}
k\left(x,u, x^{\prime},u^{\prime}\right) \triangleq k^{f}\left(x, x^{\prime}\right)+\sum_{j=1}^m u_j k^{g_{j}}\left(x, x^{\prime}\right) u_j^{\prime}.
\end{equation}
Here, \(k^f, k^{g_j}: \mathbb{R}^{n \times n} \to \mathbb{R}\) are kernels that capture the behavior of the individual entries of $f$ and $g$, respectively; $u_j$ is the $j_{\text{th}}$ dimension of the control input $u$, and $g_j$ is the $j_{\text{th}}$ column of the input matrix $g$.
\begin{remark}
Typically, GP learning addresses single-dimensional output. For a vector-valued GP dynamical model, the common approach is to define a separate GP prior for each dimension \citep{castaneda2021gaussian}. Following the conventions of previous works \cite{berkenkamp2016safe,lederer2020training}, we assume $n=1$ and $m=1$ to simplify the notation. In this case, the GP model $f+gu$ is a scalar function, and the kernel function becomes 
$k(x,u, x',u') = k^f(x, x') + u k^g(x, x')u'$.
\end{remark}
\begin{remark}
In robotic control scenarios such as quadrotor tracking, learning a GP dynamical model online is crucial, especially when dealing with varying dynamics.  For example, when quadrotors carry payloads like a bottle of water, the fluid dynamics induced by the swaying water cannot be precisely modeled and are considered disturbances, necessitating online learning. As demonstrated in our final experiment, the inputs to the GP for this water-carrying task include the quadrotor's global position, velocity and attitude, all of which are obtained through a motion capture system. The output of the GP, which serves as the label for online dynamical learning, is the disturbance force caused by fluid dynamics. This disturbance is deduced from the quadrotor’s nominal model using acceleration, thrust, and the gravity force.
\end{remark}    
For \( N \) collected measurements \( Y = \begin{bmatrix}
f(x_1) + g(x_1) u_1,\, \ldots,\, f(x_N) + g(x_N) u_N  
\end{bmatrix}^{\top} \) at state control pairs \( [X, \, U] = \begin{bmatrix}
(x_1, u_1),\, \ldots,\, (x_N, u_N))    
\end{bmatrix} \), 
the posterior mean \( \mu_{*} \) and covariance functions \( k_{*} \) are given by
\begin{subequations}\label{eq.GP perfect posterior}
\begin{align}
&\mu_*(x, u) 
\nonumber \\
=& \mu(x,u) + k(x,u,X,U) K^{-1}(Y-\mu(X, U)),\\
&k_{*}(x,u, x^{\prime},u^{\prime}) \nonumber \\
=& k(x,u, x^{\prime},u^{\prime})- k(x,u, X,U)K^{-1}k(X,U, x^{\prime},u^{\prime}).
\end{align}    
\end{subequations}
Here, the covariance matrix $K$, also known as kernel Gram matrix, is computed as $K \triangleq K_f + U^{\top}_{\mathrm{diag}} K_g U_{\mathrm{diag}}$, with $[K_f]_{(p,q)} \triangleq k^f(x_p, x_q)$, $[K_g]_{(p,q)} \triangleq k^g(x_p, x_q)$ and $U_{\mathrm{diag}} \triangleq \mathrm{diag} \left\{ u_1, \, u_2, \, ..., \, u_N\right\}$;
$k(x,u, X,U)$ is computed as $
k(x,u, X,U) \triangleq k^f(x,X) + u k^g(x,X) U_{\mathrm{diag}}$, with $[k^f(x, X)]_{(p)} = k^f(x, x_p)$, $[k^g(x, X)]_{(p)} = k^g(x, x_p)$; and $\mu(X, U)$ is defined by $[\mu(X, U)]_{(p)} = f(x_p) + g(x_p)u_p$.

In \eqref{eq.GP perfect posterior}, computing the inverse of $K$ poses a significant computational burden, scaling cubically with the number of measurements. When the computation is performed exactly, i.e., we have sufficient computational resources to compute $K^{-1}$ exactly, the modeling error is solely quantified by the quantity
$k_{*}(x,u, x^{\prime}, u^{\prime})$ in \eqref{eq.GP perfect posterior}, which corresponds to the mathematical uncertainty of the learned model \citep{williams2006gaussian,hennig2022probabilistic,wenger2022posterior}.

However, in real-world applications, computing $v_{*}=K^{-1}(Y-\mu(X,U))$ exactly is often prohibitive due to constrained computational resources, especially in online learning scenarios \cite{nguyen2008local}. Therefore, the effect of constrained computation must be considered. Analogous to how limited data induces modeling error captured by mathematical uncertainty in GP learning, constrained computation introduces an approximation error that must be accounted for as computational uncertainty \cite{hennig2022probabilistic,wenger2022posterior}. From this perspective, the GP should return a combined uncertainty that includes both mathematical and computational aspects \cite{hennig2022probabilistic,wenger2022posterior}.

As shown in \cite{wenger2022posterior}, $v_*$ can be treated as a random variable with prior distribution $v_* \sim \mathcal{N}(0, K) = \mathcal{N}(v_0, \Sigma_0)$ to include computational uncertainty. In this case, using the canonical CG method \cite{hestenes1952methods} to compute $K^{-1}$ can be regarded as performing Bayesian inference on $v_*$, with each iteration conditioning the current belief distribution $v_* \sim \mathcal{N}(v_{i-1}, \Sigma_{i-1})$ on the linear projection of the residual $s_i^{\top}r_{i-1} = s_i^{\top}(y-\mu(X,U) - Kv_{i-1}) = s_i^{\top}K(v_* - v_{i-1})$. The resulting belief update at the $i_{\text{th}}$ iteration now becomes \cite{wenger2022posterior}: 
\begin{equation}\nonumber
\begin{aligned}
v_i &= v_{i-1} + \Sigma_{i-1}Ks_i (s_i^{\top}K\Sigma_{i-1}Ks_i)^{-1}s_i^{\top}K(v_* - v_{i-1})
\\
&= C_i(y-\mu(X,U)),
\\
\Sigma_{i} &= \Sigma_{i-1} - \Sigma_{i-1}Ks_i (s_i^{\top}K\Sigma_{i-1}Ks_i)^{-1}s_i^{\top}K\Sigma_{i-1} \\
&= K^{-1} - C_i.
\end{aligned} 
\end{equation}
For the belief $v_* \sim \mathcal{N}(v_i, \Sigma_i)$, we have
\begin{equation}\nonumber
\begin{aligned}
p(f+gu) &= \int p(f+gu|v_*)p(v_*) \, \mathrm{d}v_* = \mathcal{N}(\mu_i, k_i),
\end{aligned}  
\end{equation}
with
\begin{subequations}\label{eq.mu_i, k_i}
\begin{align}
& \mu_i(x,u) \nonumber \\ 
=& \mu(x,u) + k(x,u, X,U)v_i, \\
&k_i(x,u, x^{\prime},u^{\prime}) \nonumber\\
=& \underbrace{k(x,u, x^{\prime},u^{\prime}) - k(x,u,X,U)C_ik(X,U, x^{\prime},u^{\prime})}_{\textbf{Combined Uncertainty}} \nonumber\\
=& \underbrace{k(x,u, x^{\prime},u^{\prime}) - k(x,u, X,U)K^{-1}k(X,U, x^{\prime},u^{\prime})}_{\textbf{Mathematical Uncertainty}} \nonumber\\
&+ \underbrace{k(x,u, X,U)\Sigma_ik(X,U, x^{\prime},u^{\prime})}_{\textbf{Computational Uncertainty}} \nonumber
\\
\triangleq& k^{\textup{math}}(x, u, x^{\prime}, u^{\prime}) + k^{\textup{comp}}_i(x,u,x^{\prime}, u^{\prime}).
\end{align}    
\end{subequations}

In \eqref{eq.mu_i, k_i}, $C_i$ is defined as $C_i = S_i(S_i^\top \hat{K}S_i)^{-1}S_i^\top$, which is a rank-$i$ matrix. Here, $S_i$ is given by $S_i \triangleq \begin{bmatrix} s_1 & s_2 & \cdots & s_i \end{bmatrix} \in \mathbb{R}^{N \times i}$, where each $s_i$ represents the direction of the CG update in the $i_{\text{th}}$ iteration. It satisfies the relation $s_i = \frac{r_{i-1}^\top K s_{i-1}}{s_{i-1}^\top K s_{i-1}} s_{i-1}$, where $r_{i-1}$ is the residual at the $(i-1)_{\text{th}}$ iteration. 
As we perform more computations, the uncertainty about $v_*$ contracts as $C_i \to K^{-1} = \Sigma_0$ as $i \to N$. After $N$ iterations, $C_N = K^{-1}$, the computational uncertainty reduces to zero, and the combined uncertainty reduces to the mathematical uncertainty. 

The decomposition of combined uncertainty in GPs into mathematical and computational uncertainties is already well-established in prior research \cite{hennig2022probabilistic,wenger2022posterior}. These studies primarily emphasize the benefits of leveraging computational uncertainty to enhance generalization for supervised learning tasks with medical datasets \cite{hennig2022probabilistic,wenger2022posterior}.  Different from those works, this paper focuses on robotic control. We provide a novel insight that computational uncertainty can be utilized to quantify computational errors when learning dynamical models. More importantly, in the subsequent sections, we illustrate how this computational uncertainty can further be leveraged to better estimate the region of attraction and design a stable controller in computation-constrained robotic systems.

\begin{remark}
Besides CG, other methods like Cholesky decomposition \cite{higham1990analysis,krishnamoorthy2013matrix} can also be employed for GP dynamical model learning, and the computational uncertainty can similarly be quantified \cite{wenger2022posterior,hennig2022probabilistic}.    
\end{remark}

Leveraging the convergence result of general GP models from Corollary 1 in \cite{wenger2022posterior} and combining it with Assumption~\ref{assump.bounded RKHS}, the combined uncertainty can be utilized to establish the convergence of the posterior mean of the learned dynamical GP model:

\begin{lemma}[Pointwise Convergence of the Posterior Mean]\label{lemma.worst case error}
$\forall \, [x,u] \in \mathcal{X} \times \mathcal{U}$, we have $f(x) + g(x)u \in \mathcal{H}_k$ and
\begin{align}\label{eq.pointwise convergence}
\vert f(x) + g(x)u - \mu_i(x,u) \vert &\leq B_{f,g} \cdot  \sqrt{k_i(x, u, x, u)}.
\end{align}
Here, $B_{f,g}$ is the upper bound of the RKHS norm satisfying $\|f+gu \|_{\mathcal{H}_k} \leq B_{f,g}$.
\end{lemma}
The detailed proof of Lemma~\ref{lemma.worst case error} can be found in Appendix \ref{appendix.proof of lemma 1}. Lemma~\ref{lemma.worst case error} shows that, in the sense of point-wise convergence, the combined uncertainty $k_i$ is the correct object characterizing the belief about the latent function $f+gu$ under constrained computation. 

\begin{remark}\label{remark.computation complexity}
There is an intrinsic conflict between uncertainty quantification and computational efficiency, as quantifying computational uncertainty requires additional costs. Typically, computation-aware GP model learning exhibits a quadratic time complexity of \(O(N^2i)\) for \(i\) iterations, which incurs a greater cost compared to linear-time GP approximations such as inducing point methods \cite{hensman2013gaussian,titsias2009variational}. However, in the subsequent sections, we will demonstrate that these additional costs are justifiable, as they are indispensable for providing stability guarantees in computation-constrained systems.
\end{remark}

\section{Computation-Aware Stability Analysis}\label{sec.stability analysis}

In Section~\ref{sec.model learning}, we establish the error bound in Lemma \ref{lemma.worst case error} to link constrained computation with regression performance through computational uncertainty. However, this bound does not directly address the actual objective of control systems, which is achieving stable control. This and the following sections will bridge this gap. In this section, we will first address the question:
\begin{mdframed}
\textbf{Q2:} How to verify the stability of the system using the computation-aware GP model?
\end{mdframed}
The core idea is to examine the Lyapunov function derivative along the trajectories of the system incorporating the computation-aware GP model. First, we make a common assumption about the equilibrium point of the system \cite{khalil2002nonlinear}:
\begin{assumption}[Equilibrium Point]\label{assump.equilibrium point}
The origin is an equilibrium point of \eqref{eq.system model} with $f(0)=g(0)=0$. 
\end{assumption}
Stability is crucial for robot control. The controller of a robotic system, denoted as \(\pi: \mathcal{X} \rightarrow \mathcal{U}\), must stabilize around a reference point or follow a reference trajectory. Specifically, asymptotic stability implies that if the system starts `close enough' to the equilibrium point, it remains `close enough' indefinitely and eventually converges to it:
\begin{definition}[Asymptotic Stability]
System \eqref{eq.system model} is considered asymptotically stable if, for every $\epsilon > 0$, there exists a $\delta > 0$ such that $\|x(0)\| \leq \delta$ implies $\| x(t) \| < \epsilon$ and $\lim_{t \to \infty} \|x(t)\| = 0$ for every $t \geq 0$.
\end{definition}
A common way to verify stability is to use a well-defined Lyapunov function, which can be established through natural energy functions \cite{Khalil:1173048} or experimentation \cite{giesl2015review}:
\begin{assumption}[Well-defined Lyapunov Function]\label{assump.Lyapunov function}
A Lyapunov function $V(x)$ is fixed and two-times continuously differentiable.
\end{assumption}
In this paper, we assume that such a function is known in advance. For asymptotic stability, for every $x \in \mathbb{R}^n \setminus {0}$, the derivative of the Lyapunov function is required to be less than $0$ for the closed-loop system:
\begin{equation}\label{eq.asymptotic stability}
\dot{V}(x) = \frac{\partial{V}}{\partial{x}} \left(f(x) + g(x) \pi(x) \right) < 0.    
\end{equation}
However, $\dot{V}(x) < 0$ is too stringent to achieve for all $x \neq 0$. In such cases, the stability of the system is typically assessed by the size of the region where $\dot{V}(x) < 0$. This region is proven to be an invariant set, which is specifically referred to as the region of attraction (ROA) \cite{khalil2002nonlinear}:
\begin{definition}[Region of Attraction]
A set $\mathcal{R}$ is defined as a ROA of the system if, for all $x(0) \in \mathcal{R}$, it holds that $\lim_{t \to \infty} x(t) = 0$.
\end{definition}
The concept of ROA is not only theoretically significant but also practically useful for robotic systems. It provides an essential verification of the controller's safety before its real-world deployment. For example, in drone navigation, the ROA would specify the conditions under which the drone can be expected to stabilize to a desired flight path or hover point. Finding the exact ROA analytically might be difficult or even impossible \cite{khalil2002nonlinear}. Nevertheless, we can instead use a level set of Lyapunov functions to conservatively estimate it, as shown in the following lemma:  
\begin{lemma}[Level Set Estimates of ROA  \cite{khalil2002nonlinear}]\label{lemma.ROA}
Consider a level set of the Lyapunov function $\mathcal{V}(c) = \left\{x \in \mathcal{X} |V (x) \leq c\right\}$ with $c>0$. If for every $x \in \mathcal{V}(c) \setminus \{0\}$,  condition \eqref{eq.asymptotic stability} is satisfied, then $\mathcal{V}(c)$ is an invariant set and $\mathcal{V}(c) \subseteq \mathcal{R}$.
\end{lemma}

Lemma \ref{lemma.ROA} shows under which condition the level set of a Lyapunov function $\mathcal{V}(c)$ can be verified as a subset of ROA. Consequently, the largest subset is the one most closely approximating the ROA and is often used as the ROA estimate. The goal of this section is to use the computation-aware GP model to find the largest $c$ that satisfies the condition $\dot{V}(x) < 0$ for every $x \in \mathcal{V}(c) \setminus \{0\}$. 

Using Lemma \ref{lemma.worst case error}, for a given controller $u=\pi(x)$, we have
\begin{equation} \nonumber
\begin{aligned}
&\vert f(x) + g(x)\pi(x) - \mu_i(x,\pi(x)) \vert 
\\
\leq& 
B_{f,g} \cdot  \left( \sigma^{\textup{math}}(x, \pi(x)) + \sigma^{\textup{comp}}_i(x,\pi(x)) \right),      
\end{aligned}
\end{equation}
where $B_{f,g}$ is the upper bound of the RKHS norm satisfying $\|f+gu \|_{\mathcal{H}_k} \leq B_{f,g}$;
$\sigma^{\textup{math}}(x,u) \triangleq \sqrt{k^{\textup{math}}(x,u,x,u)}$ and $\sigma_i^{\textup{comp}}(x,u) \triangleq \sqrt{k_i^{\textup{comp}}(x,u,x,u)}$ are the standard deviation of the GP posterior at $(x,u)$.
Thus, the derivative of the Lyapunov function $\dot{V}(x)$ can be bounded by:
\begin{equation} \nonumber
\begin{aligned}
\dot{V}(x) 
=& \frac{\partial V(x)}{\partial x} \left( f(x)+g(x)\pi(x) \right)\\
=& \frac{\partial V(x)}{\partial x} \left( \mu_i(x,\pi(x)) + f(x)+g(x)\pi(x) - \mu_i(x,\pi(x)) 
 \right)
\\
\leq& \frac{\partial V(x)}{\partial x} \mu_i(x,\pi(x)) \\
&+ \left \vert \frac{\partial V(x)}{\partial x}  
\right \vert \left \vert f(x) + g(x)\pi(x) - \mu_i(x,\pi(x)) \right \vert
\\
\leq& \underbrace{\frac{\partial V(x)}{\partial x} \mu_i(x,\pi(x))}_{\triangleq \dot{\bar{V}}_i(x)} + 
\underbrace{  \left \vert \frac{\partial V(x)}{\partial x} \right \vert \cdot B_{f,g} \cdot \sigma^{\textup{math}}(x, \pi(x))}_{\triangleq \dot{V}^{\textup{math}}(x)} 
\\
&+ \underbrace{\left \vert \frac{\partial V(x)}{\partial x} \right \vert \cdot B_{f,g} \cdot   \sigma^{\textup{comp}}_i(x,\pi(x))}_{\dot{V}_i^{\textup{comp}}(x)}.
\end{aligned}
\end{equation}

From the derivation above, we find that $\dot{V}(x)$ can be upper bounded by the sum of $\dot{\bar{V}}_i(x)$, $\dot{V}_i^{\textup{comp}}(x)$, and $\dot{V}^{\textup{math}}(x)$. Here, $\dot{\bar{V}}_i(x)$ can be considered the mean part of the Lyapunov derivative, while $\dot{V}^{\textup{math}}(x)$ relates to mathematical uncertainty, and $\dot{V}_i^{\textup{comp}}(x)$ pertains to computational uncertainty. 
This decomposition of
$\dot{V}(x)$ allows us to estimate the ROA using the computation-aware GP model:
\begin{theorem}[ROA Estimation using Computation-Aware GP Model]\label{theorem.ROA}
If, for all $x \in \mathcal{V}(c) \setminus \{0\}$, the following inequality is satisfied:
\begin{equation}\label{eq.sufficient condition}
\dot{\hat{V}}(x) \triangleq \dot{\bar{V}}_i(x) + \dot{V}_i^{\textup{comp}}(x) + \dot{V}^{\textup{math}}(x) < 0,
\end{equation}
then it follows that $\mathcal{V}(c) \subseteq \mathcal{R}$.
\end{theorem}
\begin{proof}
Note that \eqref{eq.sufficient condition} represents a sufficient condition for \eqref{eq.asymptotic stability}. Therefore, this theorem can be directly derived from Lemma \ref{lemma.ROA}.
\end{proof}
This theorem offers a sufficient condition for estimating the ROA using a computation-aware GP model. The decomposition of $\dot{V}(x)$ indicates that both mathematical and computational uncertainties are indispensable in estimating the ROA. With a fixed amount of collected data, the mathematical uncertainty part, $\dot{V}^{\textup{math}}(x)$, remains constant, while the computational uncertainty part, $\dot{V}^{\textup{comp}}_i(x)$, typically decreases as more computations are performed \cite{wenger2022posterior, hennig2022probabilistic}.
\begin{remark}
Previous ROA estimation methods that use GP dynamic model learning, such as those described in \cite{berkenkamp2016safe}, do not account for computational uncertainty. This oversight can lead to estimated ROAs that lack stability guarantees, meaning the estimated level sets of the Lyapunov function cannot assure they are subsets of the ROA under constrained computation. For instance, it has been observed that the ROAs estimated using the method in \cite{berkenkamp2016safe} exceed the actual ROAs in scenarios with computational constraints, as demonstrated in Fig.~\ref{fig.2D comparision}.
\end{remark}

In practice, it is often impossible to evaluate $\dot{V}(x)$ everywhere in a continuous domain. Nonetheless, the continuity of $\dot{V}$ allows for evaluating it only at a finite number of points without losing guarantees, as shown in the subsequent corollary:

\begin{corollary}[ROA Estimation Through Discretization]\label{corollary.discretization}
Assume that the given policy $\pi$ is bounded in $\mathcal{X}$ satisfying $\|\pi\|_{\infty} \leq B_{\pi}$. If, for every $x \in \left\{ \mathcal{X}_{\tau} \cap \left\{\mathcal{V}(c) \setminus {0}\right\} \right\}$, the following inequality is satisfied:
\begin{equation}\label{eq.sufficient condition discretization}
\dot{\bar{V}}_i(x) + \dot{V}^{\textup{math}}(x) + \dot{V}_i^{\textup{comp}}(x) < -L\tau,  
\end{equation}
where $L$ is the Lipschitz constant of $\dot{V}(x)$ defined by
\begin{equation}\nonumber
\begin{aligned}
L =& \left(B_f^2 \|k^f\|_{\infty} + B_g^2 B_{\pi} \|k^g\|_{\infty}  \right) \cdot \left\|\frac{\partial^2 E(x)}{\partial x^2}\right\|_{\infty}
\\
+& \sqrt{2}\left\|\frac{\partial E(x)}{\partial x}\right\|_{\infty}\cdot \left(  B_f \sqrt{ \|k^f\|_{\infty}\left\|\frac{\partial k^f}{\partial x}\right\|_{\infty}} \right)
\\
+& \sqrt{2}\left\|\frac{\partial E(x)}{\partial x}\right\|_{\infty}\cdot \left(  B_g \pi_g  \sqrt{ \|k^g \|_{\infty}\left\|\frac{\partial k^g}{\partial x}\right\|_{\infty}} \right).     
\end{aligned}   
\end{equation}
It then follows that $\mathcal{V}(c) \subseteq \mathcal{R}$.
Here, $\mathcal{X}_{\tau} \in \mathcal{X}$ is a discretization of $\mathcal{X}$ satisfying $|x - [x]_{\tau}| \leq \frac{\tau}{2}$ for all $x \in \mathcal{X}$; $[x]_{\tau}$ defines the nearest point in $\mathcal{X}_{\tau}$ to $x \in \mathcal{X}$. 
\end{corollary}
The detailed proof of Corollary~\ref{corollary.discretization} can be found in Appendix~\ref{appendix.proof of corollary 1}. The Lipschitz continuity of $\dot{V}(x)$ allows for the extension of results from $\mathcal{X}_{\tau}$ to the broader set $\mathcal{X}$. Additionally, the trade-off between accuracy and computational effort can be managed by adjusting the discretization factor $\tau$. A smaller value of $\tau$ leads to a less conservative estimate of the ROA. 

This corollary provides a way for estimating the ROA by means of binary search. For
high-dimensional systems, this approach may become computationally demanding due to
the exponential increase in the number of points to be checked with each added dimension.
Nevertheless, the primary purpose of quantifying the ROA is to evaluate the stability of a given
controller before deployment. Therefore, it is typically performed offline, and the required
time could be manageable for offline computation.    

\section{Computation-Aware Controller Design}\label{sec.controller design}
Section \ref{sec.stability analysis} introduces the method to verify the stability of a given controller. 
If a given controller does not meet the desired ROA, it needs to be redesigned using the computation-aware GP model.
This section answers the following question: 
\begin{mdframed}
\textbf{Q3}
How can a stable controller be synthesized for a computation-constrained system? 
\end{mdframed}
The core idea involves formulating an optimization problem that minimizes the control input norm, subject to the inequality constraints $\dot{V} < 0$, by employing a computation-aware GP model. 
We begin with the definition of the control Lyapunov function (CLF):
\begin{definition}[Control Lyapunov Function]
The function $V(x)$ is a control Lyapunov function for the system \eqref{eq.system model} if, for each $x \in \mathbb{R}^n \setminus{0}$, it holds that
\begin{equation} \nonumber
\label{eq.CLF}
\inf _{u \in \mathbb{R}^m} \underbrace{\frac{\partial V(x)}{\partial x}f(x) +\frac{\partial V(x)}{\partial x}g(x)u}_{\triangleq\dot{V}(x, u)}<0. 
\end{equation}
\end{definition}
The existence of such a CLF indicates that the system is globally stabilizable if $V$ is radially unbounded, i.e., $V(x) \to \infty$ as $|x| \to \infty$ \cite{artstein1983stabilization}. The objective then becomes to find a local Lipschitz continuous feedback control law $u = \pi(x)$, ensuring that the condition $\dot{V}(x, \pi(x)) < 0$ is met for all $x \in \mathbb{R}^n \setminus{0}$. For unconstrained inputs $u \in \mathbb{R}^m$, such a feedback control law can be readily derived in a closed form \cite{sontag1989universal}. However, robotic systems typically face actuator limitations that manifest as control constraints 
 $u \in \mathcal{U}$. Here, the set $\mathcal{U} \subseteq \mathbb{R}^m$ represents the compact set of all admissible control inputs. In this case, there is no exact closed-form to construct a stable control policy and it may be unfeasible to find a control input ensuring $\dot{V}(x,u) < 0$ for every $x \in \mathbb{R}^n \setminus{0}$, even if $V$ is a valid CLF for the system.
In such cases, it is possible to identify a level set of the Lyapunov function $\mathcal{V}(c) = \left\{x \in \mathcal{X} |V (x) \leq c\right\}$, where
\begin{equation} \nonumber
\inf _{u \in \mathcal{U}} \frac{\partial V(x)}{\partial x}f(x) +\frac{\partial V(x)}{\partial x}g(x)u<0,
\end{equation}
is valid for all $x\in \mathcal{V}(c) \setminus \{0\}$. 
To facilitate the convexification of the optimization problem, and also impose a stronger notion of stabilizability for exponential convergence \cite{freeman1996control,castaneda2021gaussian,cosner2023robust}, we can instead consider:
\begin{equation}\label{eq.CLF convex}
\inf _{u \in \mathcal{U}} \frac{\partial V(x)}{\partial x}f(x) +\frac{\partial V(x)}{\partial x}g(x)u + \lambda \cdot V(x) \leq 0,    
\end{equation}
with $\lambda>0$.
A straightforward method to synthesize such a control law is by enforcing \eqref{eq.CLF convex} as a constraint in a min-norm optimization problem:

\HRule
\noindent \textbf{CLF-QP}:
\begin{subequations}
\label{eq.clf-qp}
\begin{align}
\pi^*(x) & = & & \underset{u\in \mathcal{U}}{\argmin} \quad \|u \|^2 \label{eq.clf-qp1},\\
\text{s.t.} & \; & & \frac{\partial V(x)}{\partial x}f(x) +\frac{\partial V(x)}{\partial x}g(x)u + \lambda \cdot V(x) \leq 0.\label{eq.clf-qp2}
\end{align}
\end{subequations}
\HRule

In this paper, we assume that the input constraints are linear, which renders problem \eqref{eq.clf-qp} a quadratic program (QP). This optimization problem formulates a feedback control law, $\pi^*(x)$, that selects the min-norm input to ensure exponential convergence of the system state to the origin. However, this QP is not always feasible as the actual model inevitably contains uncertainties. Therefore, a principled solution involves reformulating the min-norm stabilizing controller, as defined in \eqref{eq.clf-qp}, to accommodate model uncertainty. Following \eqref{eq.sufficient condition} from Theorem \ref{theorem.ROA}, a sufficient condition for \eqref{eq.clf-qp2} is given by:
\begin{equation}\label{eq.clf sufficient condition}
\begin{aligned}
&\frac{\partial V(x)}{\partial x} \cdot \mu_i(x, u) + B_{f,g} \cdot  \left| \frac{\partial V(x)}{\partial x}\right| \cdot \sigma_i^{\textup{comp}}(x, u) 
\\
&+ B_{f,g} \cdot  \left| \frac{\partial V(x)}{\partial x}\right| \cdot \sigma^{\textup{math}}(x, u) 
 + \lambda V(x) \leq 0. 
\end{aligned}
\end{equation}
Therefore, we can replace the constraint in \eqref{eq.clf-qp2} with its sufficient condition \eqref{eq.clf sufficient condition}:

\HRule
\noindent \textbf{Computation-Aware GP-CLF-SOCP}:
\begin{equation}\label{eq.clf-SOCP}
\begin{aligned}
{\pi}^{*}(x) & = & & \underset{u \in \mathcal{U}}{\argmin} \quad \|u\|^2 
\\
\text{s.t.} & \; & & \eqref{eq.clf sufficient condition}.
\end{aligned}
\end{equation}
\HRule
The optimization problem \eqref{eq.clf-SOCP} is classified as a second-order cone program (SOCP). This arises from the inherent characteristics of the composite kernels. Specifically, \(\mu_i\) is a linear function of \(u\), while  \(\left(\sigma_i^{\textup{comp}}\right)^2\) and \(\left(\sigma^{\textup{math}}\right)^2\) are both positive definite quadratic functions of \(u\), as elucidated in \cite{castaneda2021gaussian}.

\begin{theorem}\label{theorem.clf}
If there exists a constant $c > 0$ such that the solution to the optimization problem \eqref{eq.clf-SOCP} exists for all $x \in \mathcal{V}(c) \setminus {0}$, then $\mathcal{V}(c)$ is a subset of ROA.
\end{theorem}
\begin{proof}
The proof of ROA is straightforward: for all $x$ at the boundary of $\mathcal{V}(c)$, we can always find a control input $u$ such that $\dot{V}(x, u)<0$ ensured by \eqref{eq.clf sufficient condition}.
\end{proof}
Theorem \ref{theorem.clf} ensures the stability of the controller synthesized through solving \eqref{eq.clf-SOCP}. We want to emphasize that the constraints in the optimization problem often neglect computation uncertainty \cite{castaneda2021gaussian}, i.e., \(\sigma_i^{\textup{comp}}(x,u) \equiv 0\) is implicitly assumed. Under this assumption, \eqref{eq.clf sufficient condition} no longer serves as a sufficient condition for \eqref{eq.clf-qp2}. 
\begin{remark}\label{remark.constraints}
In practice, the constraint \eqref{eq.clf sufficient condition} can be relaxed by introducing a slack variable, ensuring the feasibility of the problem in cases where it is not locally satisfied. Additionally, instead of minimizing the norm of the control signals, we often minimize their weighted norm to normalize the impact of the magnitude of each control quantity. Consequently, the optimization problem in \eqref{eq.clf sufficient condition} can be reformulated as
\begin{equation}
\begin{aligned}
{\pi}^{*}(x) & = & & \underset{u \in \mathcal{U},\ d \in \mathbb{R}}{\argmin} \quad \|u\|^2_{W} + p \, d^2 \label{eq:gp-clf-socp} \\
\emph{\text{s.t.}} & \; & & \frac{\partial V(x)}{\partial x} \cdot \mu_i(x, u) + B_{f,g} \cdot  \left| \frac{\partial V(x)}{\partial x}\right| \cdot \sigma_i^{\textup{comp}}(x, u) 
\\
 & & & + B_{f,g} \cdot  \left| \frac{\partial V(x)}{\partial x}\right| \cdot \sigma^{\textup{math}}(x, u) 
 + \lambda V(x) \leq d
\nonumber.
\end{aligned}
\end{equation}
Here, $\Vert u \Vert_{W}^2 \triangleq u^{\top}Wu$.
\end{remark}
In a specific scenario where only the drift dynamics \( f(x) \) is unknown, and the control matrix \( g(x) \) is known, we need to learn only \( f(x) \). In this context, the posterior mean for learning \( f(x) \) is denoted as \( \mu_i(x) \), while the uncertainties are expressed as \( \sigma_i^{\textup{comp}}(x) \) and \( \sigma^{\textup{math}}(x) \), respectively. Under mild assumptions, it is possible to directly construct an explicit form of a stable controller for unconstrained inputs without the necessity of solving an optimization problem:
\begin{equation}\label{eq.explicit policy}
\begin{aligned}
\pi^{\mathrm{explicit}}(x) =
\displaystyle -\frac{a(x)+\sqrt{a^2(x)+\|b(x)\|^4}}{\|b(x)\|^2} b(x).
\end{aligned}
\end{equation}
Here, $a(x)$ is a scalar function defined as
\begin{equation}\label{eq.ax}
\begin{aligned}
a(x)\triangleq \frac{\partial{V(x)}}{\partial{x}}\mu_i(x)+ B_{f} \cdot  \left| \frac{\partial V(x)}{\partial x}\right| \cdot \left( \sigma_i^{\textup{comp}}(x) + \sigma^{\textup{math}}(x) \right),    
\end{aligned}
\end{equation}
and $b(x)$ is a vector function defined by
\begin{equation}\label{eq.bx}
b(x)\triangleq g^{\top}(x)\frac{\partial{V(x)}}{\partial{x^{\top}}}.    
\end{equation}
The stability of this explicit control policy is summarized in the subsequent Proposition:
\begin{proposition}
[Explict Form of Stable Control Policies]\label{prop.explicit form}
Assuming that the input matrix $g(x) \neq 0$ and $\frac{\partial{V(x)}}{\partial{x}} \neq 0$ for all $x$. Then, the explicit control policy in \eqref{eq.explicit policy} stabilizes the closed-loop system.
\end{proposition}
\begin{proof}
Defining $a(x)$ and $b(x)$ as in \eqref{eq.ax} and \eqref{eq.bx}, we have
\begin{equation} \nonumber
\begin{aligned}
&\frac{\partial{V(x)}}{\partial{x}}\mu_i(x)+ B_{f} \cdot  \left| \frac{\partial V(x)}{\partial x}\right| \cdot \left( \sigma_i^{\textup{comp}}(x) + \sigma^{\textup{math}}(x) \right) 
\\
+& \frac{\partial V(x)}{\partial x}g(x)\pi^{\mathrm{explicit}}(x) \\
=& a(x) + b^{\top}(x) \left(-\frac{a(x)+\sqrt{a^2(x)+\|b(x)\|^4}}{\|b(x)\|^2} b(x)\right)
\\
=& -\sqrt{a^2(x)+\|b(x)\|^4} < 0.
\end{aligned}
\end{equation}
\end{proof}

The control policy described in \eqref{eq.explicit policy} can be viewed as an extension of Sontag's universal formula \cite{lin1991universal}. The key difference lies in our theorem's additional consideration of the model's mathematical and computational uncertainties. If the model learning is sufficiently accurate, for example, using sufficient data and ample computation, our approach reduces to Sontag's universal formula.
\begin{remark}
We have rigorously proven that by using a computation-aware GP model, we can synthesize a stable controller in computation-constrained systems. However, there are certain trade-offs involved. Intuitively, when computational uncertainty is present, the complexity of model learning increases, as noted in remark \ref{remark.computation complexity}. However, this increase in complexity is usually not critical and can often be disregarded \cite{wenger2022posterior}. Meanwhile, compared to approaches that do not account for computational uncertainty, the constraints in solving the min-norm optimization problem are tighter, often leading to larger control signals.
\end{remark}
\section{Simulations}
In this section, we consider two simulations of the canonical control tasks to show the effectiveness of computation-aware learning, computation-aware stability analysis, and computation-aware controller design.
\subsection{Nonlinear 1D System}
Consider a canonical 1D system \cite{berkenkamp2016safe} described by 
\begin{equation}\label{eq.1D system model}
\dot{x} = f(x) + u,    
\end{equation}
where $x \in \mathbb{R}$ is the system state, and the system dynamics $f(x)$ is sampled from a GP with zero mean and a composite kernel, i.e., $f \sim \mathcal{GP}(0, k)$. To achieve a reasonably good performance, the kernel $k$ is chosen as a product of a linear kernel and a Matérn kernel, defined as $k\left(x, x^{\prime}\right) \triangleq k_{\text{linear}}\left(x, x^{\prime}\right) * k_{\text{Matérn}}\left(x, x^{\prime}\right)$ \cite{berkenkamp2016safe}.

\begin{figure*}[htp]
    \centering
    \includegraphics[width=1.0\linewidth]{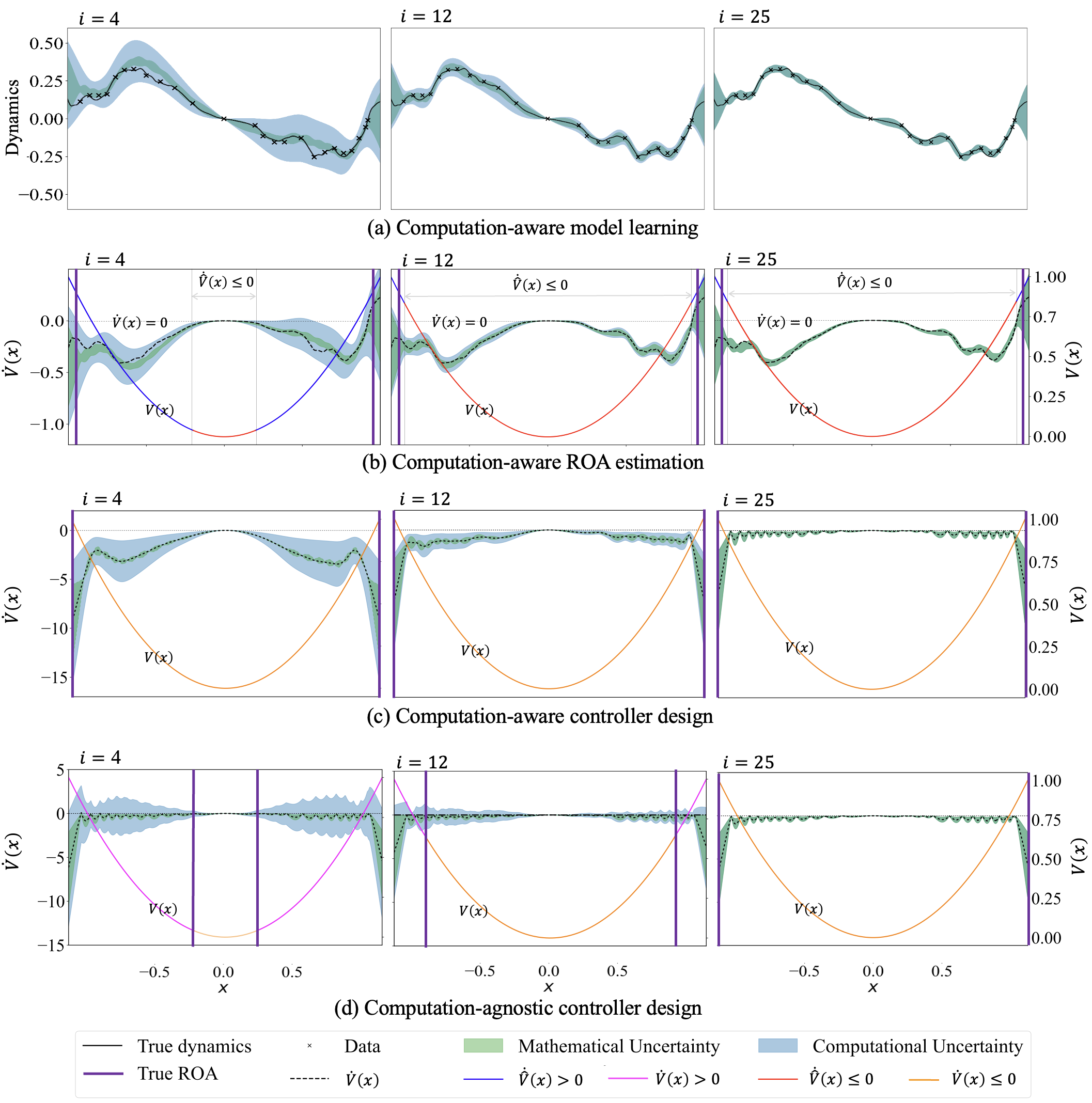}
 \caption{Comprehensive evaluation of model learning and stable control in a nonlinear 1D System \eqref{eq.1D system model}. During the GP model training, we systematically gather 25 data points per online update and vary the CG iterations for model learning (\(i=4\), \(i=12\), \(i=25\)). 
  The combined uncertainty of the learned dynamics, Lyapunov function $V(x)$ and its derivative $\dot{V}(x)$ in~\eqref{eq.asymptotic stability} respectively, decompose into the mathematical uncertainty 
  (\mysquare[mplblue!40]) and computational uncertainty (\mysquare[mplgreen!40]).
  (a) Computation-aware model learning. The computational uncertainty diminishes with increasing \(i\), while mathematical uncertainty remains constant. (b) Computation-aware ROA estimation. The estimated ROA grows with each iteration, converging towards the true ROA using a given control policy \(\pi(x)=-2.5  x\). (c) Computation-aware controller design. Our method \eqref{eq.clf-SOCP} always ensures a negative \(\dot{V}(x)\), signifying stable control, which is impressive even when the model is learned with only 4 iterations. (d) Computation-agnostic controller design. This design (setting $\sigma_i^{\textup{comp}} \equiv 0$ in \eqref{eq.clf-SOCP}), which omits computational uncertainty, yields a much smaller ROA compared to the counterparts in (c), potentially leading to unsafety in regions where \(\dot{V}(x)\) is positive.}
    \label{fig.1D_example}
\end{figure*}

\begin{figure}[!t]
    \centering
    \includegraphics[width=0.99\linewidth]{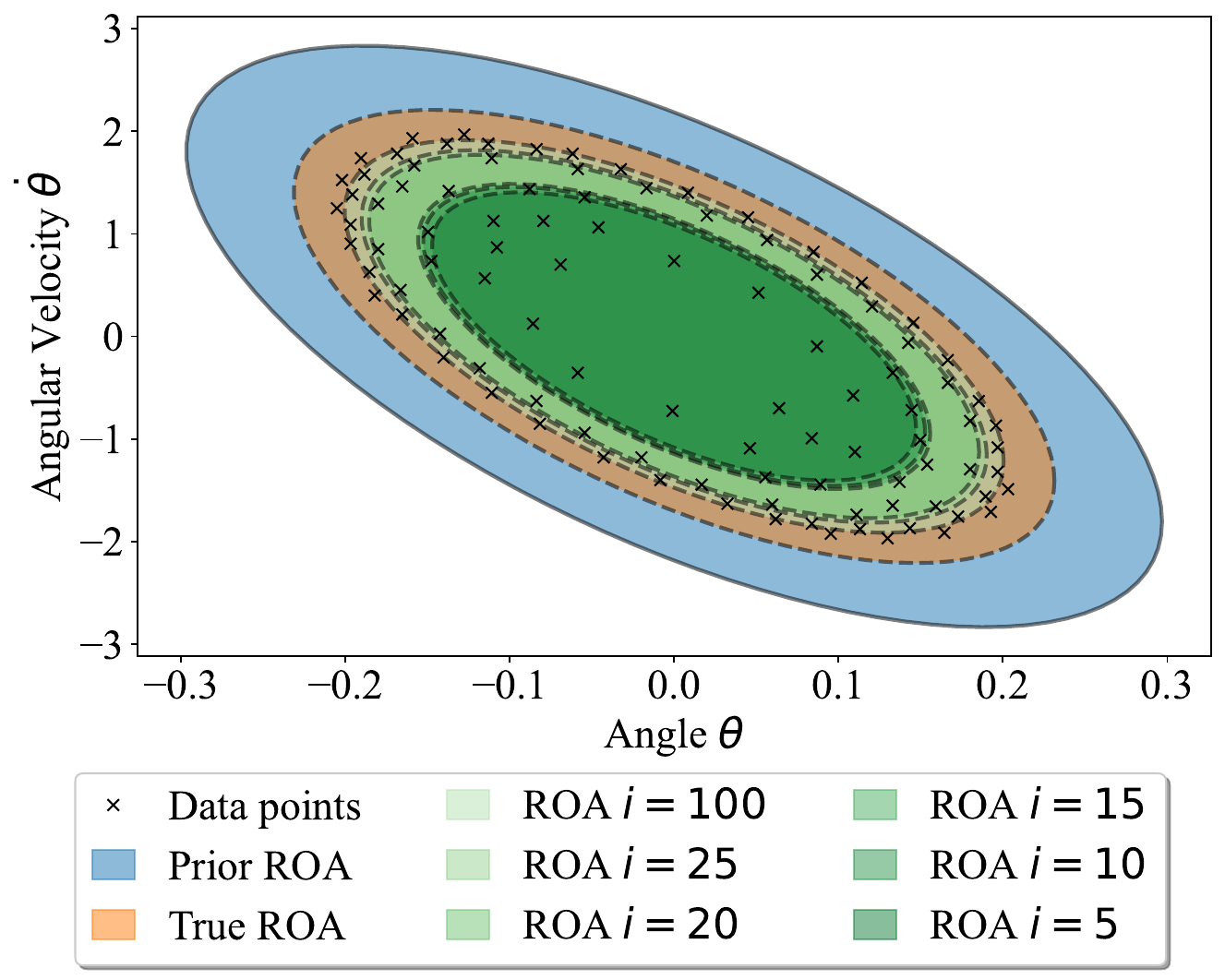}
    \caption{Computation-aware ROA estimation for inverted pendulum system given the linear quadratic controller. The ROA estimate using the prior model exceeds the true ROA. In contrast, the ROA estimates
using computation-aware GP models do not exceed the true ROA and increase with
the number of CG iterations. }
    \label{fig.2D}
    \vspace{-0.5cm}
\end{figure}

 As demonstrated in Fig.~\ref{fig.1D_example}, we decompose the uncertainty of the learned GP model into two types: mathematical uncertainty and computational uncertainty. The mathematical uncertainty of the GP dynamical model remains constant since the number of acquired data points is fixed. In contrast, the computational uncertainty decreases as the number of iterations increases. Notably, when the number of iterations equals the number of data points, we achieve ideal computation, and the computational uncertainty is reduced to zero. Furthermore, the true model is encompassed within the error bar of the combined uncertainty but exceeds that of the mathematical uncertainty alone. This indicates that the combined uncertainty provides a tight worst-case bound on the relative error between the posterior mean of the GP dynamical model and the true system. It also suggests that relying solely on mathematical uncertainty is inadequate for effective model error quantification.

Building on this computation-aware GP model, we quantify the ROA using Theorem \ref{theorem.ROA}, adhering to the fixed feedback control policy $\pi(x) = -2.5x$, as established in \cite{berkenkamp2016safe}. As shown in Fig.~\ref{fig.1D_example}, we observe that the ROA estimate increases as the computational uncertainty decreases, but never exceeds the true ROA, which is within the interval $[-0.921, 0.921]$. This observation implies that our ROA estimates are conservative approximations of the actual ROA. Additionally, we note that the ROA does not change linearly or uniformly with the number of iterations. For example, in this simple one-dimensional system, the ROA determined with 12 iterations is the interval $[-0.915, 0.915]$, and for 25 iterations, it becomes $[-0.918, 0.918]$. This suggests that a computation with 12 iterations is sufficient to accurately estimate the ROA.

Finally, we implement our proposed explicit control policy \eqref{eq.explicit policy} to construct a stable controller using the GP model acquired by 4, 12, and 25 CG iterations. For comparison, we use a similar control policy form but assume that the computational uncertainty is zero. To better evaluate their stability, we directly depict the true ROA using the real system model instead of estimating it with the GP model. Fig.~ \ref{fig.1D_example} clearly shows that the computation-aware controller achieves stability throughout the interval of $[-1,1]$. On the contrary, the controller without computational awareness leads to a much smaller ROA. For example, in the case of 4 CG iterations, the system is stable only within the interval of $[-0.2, 0.2]$. This implies that computational uncertainty is indispensable for designing a stable controller.

\subsection{Inverted Pendulum}
We further investigate the ROA estimation using the canonical inverted pendulum system \cite{berkenkamp2016safe}:
\begin{equation}\label{eq.pendulum dynamics}
\begin{aligned}
\begin{bmatrix}
\dot{\theta}(t) \\
\ddot{\theta}(t)
\end{bmatrix}
=
\begin{bmatrix}
\dot{\theta}(t) \\
\frac{mgl \sin(\theta(t)) - \mu\dot{\theta}(t)}{ml^2}
\end{bmatrix} + \begin{bmatrix}
0 \\
\frac{1}{ml^2}
\end{bmatrix}u(t)
.
\end{aligned}    
\end{equation}
Here, $\theta$ represents the angle, $m = 0.15 \mathrm{kg}$ is the mass, $l=0.5 \mathrm{m}$ is the length, and $\mu=0.05$ is the friction coefficient, which are consistent with \cite{berkenkamp2016safe}.

We assume that our knowledge is limited to a linear approximation of the dynamics \eqref{eq.pendulum dynamics} at the upright equilibrium point. In this approximation, we neglect friction and consider the mass to be $0.05\mathrm{kg}$ lighter. This linearized model serves as the nominal model and, thus, as the prior model for GP model learning. For estimating the ROA, we employ a linear quadratic regulator (LQR) based on the nominal model as the baseline controller, which uses the same weighted matrices as in \cite{berkenkamp2016safe}. In addition, the level set of its quadratic Lyapunov function is leveraged to estimate the ROA.

In this study, we collect 100 data points and present the ROA estimates for various numbers of CG 
iterations. As shown in Fig.~\ref{fig.2D}, the ROA estimate using the prior model is excessively large, exceeding the true value. In contrast, the ROA estimates using computation-aware GP models are more conservative. These estimates do not exceed the true ROA and increase with the number of CG iterations. Furthermore, to demonstrate the critical role of computational uncertainty in ROA estimates for computation-constrained systems, we compare our approach with a computation-agnostic method \cite{berkenkamp2016safe}, which overlooks computational uncertainty in estimating ROA.
From Fig.~\ref{fig.2D comparision}, we observe that for 5 CG iterations, the ROA estimate of the computation-agnostic method can exceed the true value. This suggests that neglecting computational uncertainty can lead to overly optimistic estimates of a system's stability, potentially misclassifying unstable regions as stable.
\begin{figure}[!t]
    \centering
    \includegraphics[width=0.98\linewidth]{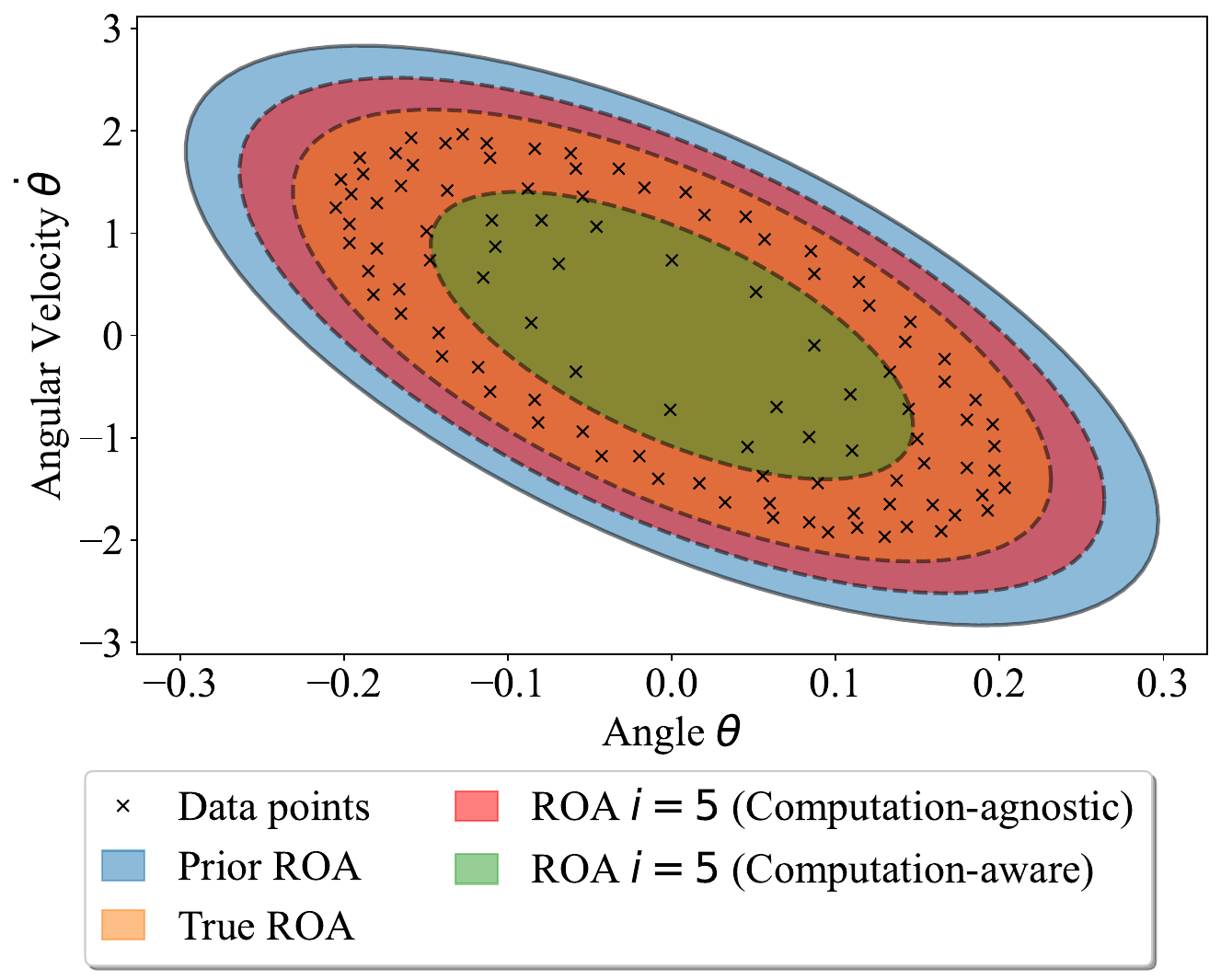}
    \caption{Comparison of computation-aware (Our method) and computation-agnostic \cite{berkenkamp2016safe} ROA estimation for inverted pendulum system under computation constraints. Neglecting computational uncertainty can lead to
overly optimistic estimates of a system’s stability, potentially
misclassifying unstable regions (\mysquare[red]) as stable.}
    \label{fig.2D comparision}
\end{figure}
\section{Experiment on Quadrotor Tracking
}\label{exp:quad}
In this section, we present an experiment on a quadrotor tracking task with a payload to evaluate the effectiveness of our proposed framework. The experimental setup is
composed of a motion capture system with 6 cameras, a
WiFi router for communication and a customized quadrotor with a flight controller Pixhawk 2.4.8 and an onboard Linux computer (Jetson Orin NX 16GB RAM). We retrofitted the quadrotor with 4 reflective infrared markers to acquire accurate
position and attitude using the motion capture system, see Fig.~\ref{fig.quadrtor}. Additionally, the accurate velocity and acceleration of the quadrotor can be obtained from the positional differences.

\begin{figure}
    \centering
    \includegraphics[width=0.5\linewidth]{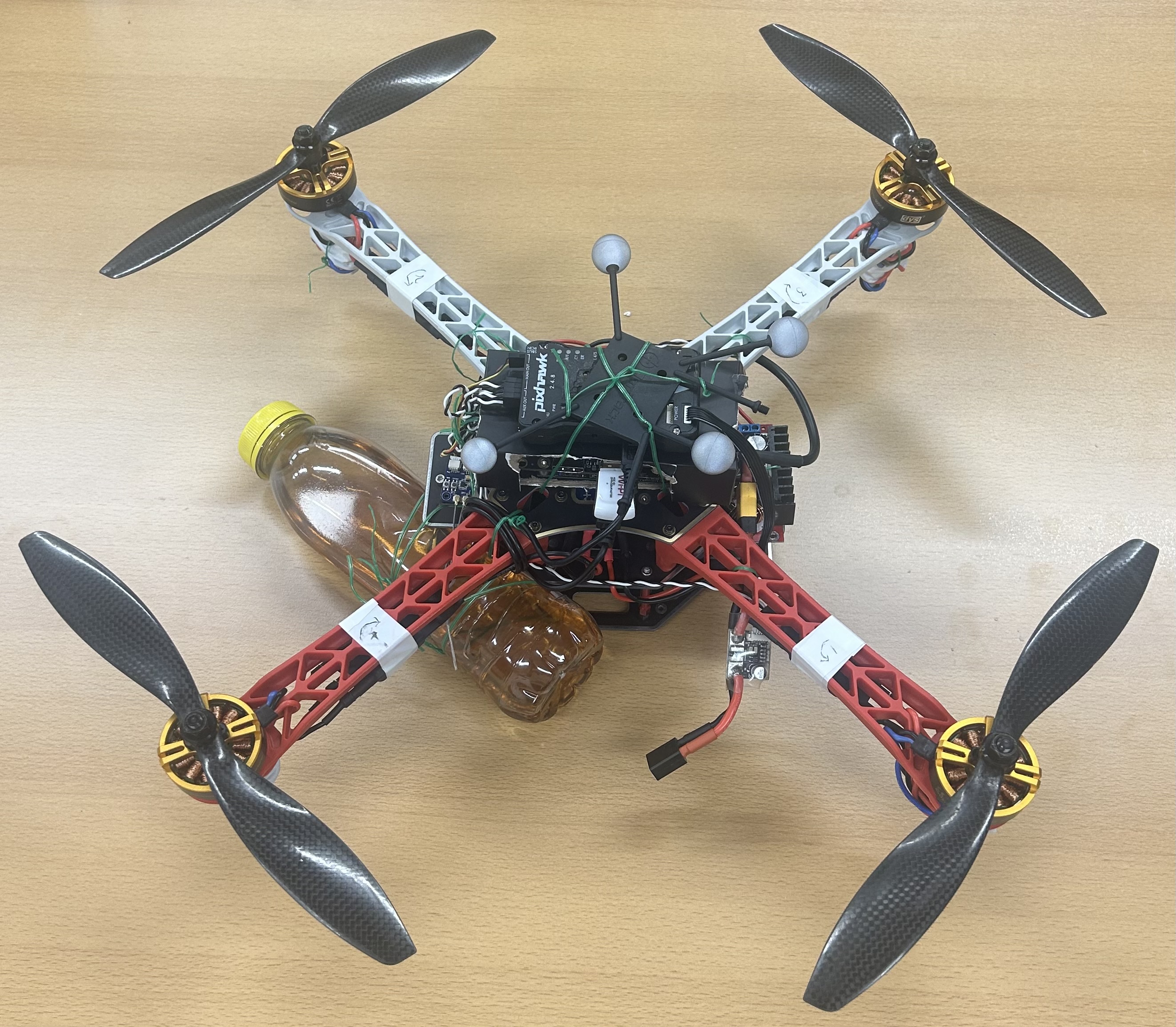}
    \caption{Customized quadrotor equipped with a Pixhawk 2.4.8 flight controller and a Jetson Orin NX onboard computer, carrying a bottle of water.}
\label{fig.quadrtor}
\end{figure}

Considering the system dynamics of a quadrotor model \cite{shi2019neural}:
\begin{subequations}
\begin{align}
\dot{p} &= v, &  
m\dot{v} &=m{g_v}+R{f}_u + {f}_d,\label{eq.quadrotor position} \\ 
\dot{R}&=RS(\omega), & 
J\dot{\omega} &= J \omega \times \omega  + \tau_u + \tau_d,\label{eq.quadrotor attitude}
\end{align}
\end{subequations}
Here, \eqref{eq.quadrotor position} and \eqref{eq.quadrotor attitude} describe the position and attitude dynamics of the quadrotor, respectively. The variables $p \triangleq \left[p_x,\, p_y,\, p_z \right]^{\top} \in \mathbb{R}^3$, $v\in \mathbb{R}^3$, $R \in \mathrm{SO}(3)$ and ${\omega} \in \mathbb{R}^3$  represent the global position, velocity, attitude rotation matrix, and body angular velocity, respectively. The symbols $m$ and $J$ denote the mass and inertia matrix, while $S$  denotes the skew-symmetric mapping. $g_v \triangleq \left[0,\, 0,\, -g\right]^{\top}$ is the gravity vector; ${f}_u = \left[0,\, 0,\, T\right]^\top$ and ${\tau}_u = \left[\tau_x,\, \tau_y,\, \tau_z\right]^\top$ are the
total thrust and body torques. 

In our experiment, we consider the task where a quadrotor carries a bottle of water and attempts to follow a single continuous trajectory, as shown in Fig.~\ref{fig.intro}. Here, $f_d$ and $\tau_d$ represent the unknown disturbance forces and disturbance torques, respectively. These disturbances are originated from the complex slosh dynamics \cite{ibrahim2005liquid} and the asymmetric centroid caused by the swaying water in the bottle, which are difficult to model accurately. To address the disturbance forces $f_d$, we use GP to learn them online. Following \cite{shi2019neural}, we employ a different, highly robust controller to manage the attitude dynamics, thus the disturbance torque $\tau_d$ is not our concern. 
The frequency of the position control loop \eqref{eq.quadrotor position} is set to 10 Hz to accommodate online GP learning and controller synthesis, whereas the frequency of the attitude control loop \eqref{eq.quadrotor attitude} is 250 Hz, which is the default setting in the Pixhawk flight controller.

To fully predict the disturbance forces \(f_d\), we use position \(p\), velocity \(v\) and rotation \(R\) as inputs to the GP model. Additionally, we employ the relationship \(f_d = m\dot{v} - mg_v - Rf_u\) to calculate \(f_d\), which serves as the measurement for the GP model. For every set of 20 data points, we perform online learning for both the computation-aware and computation-agnostic GP models using $i=5$, $i=10$, and $i=15$ CG iterations, respectively. For trajectory tracking, we define a sliding  variable $s$, which is a manifold on which the tracking error $\tilde{p} \triangleq p - p_d$ converges to zero exponentially:
\begin{equation}\nonumber
s \triangleq \dot{\tilde{p}} + \Lambda \tilde{p} = \dot{{p}} - \dot{{p}}_d + \Lambda \tilde{p},  
\end{equation}
where $\Lambda$ is a positive definite diagonal matrix and $p_d$ is the desired trajectory. Following \cite{shi2019neural}, we use the sliding variable to construct the Lyapunov function $V(s) = \frac{1}{2}m \|s\|^2$. Taking the time derivative on $V(s)$, we have
\begin{equation}\nonumber
\begin{aligned}
\dot{V}(s) &= ms^{\top} \dot{s} \\
&=ms^{\top} \left( 
\ddot{p} - \ddot{p}_d + \Lambda \dot{\Tilde{p}} 
\right)
\\
&=s^{\top} \left( 
m{g_v}+R{f}_u + {f}_d - m\ddot{p}_d + \Lambda m \dot{\Tilde{p}} 
\right).
\end{aligned}
\end{equation}

\begin{figure}[t]
    \centering
    \begin{minipage}{\linewidth}
        \includegraphics[width=\linewidth]{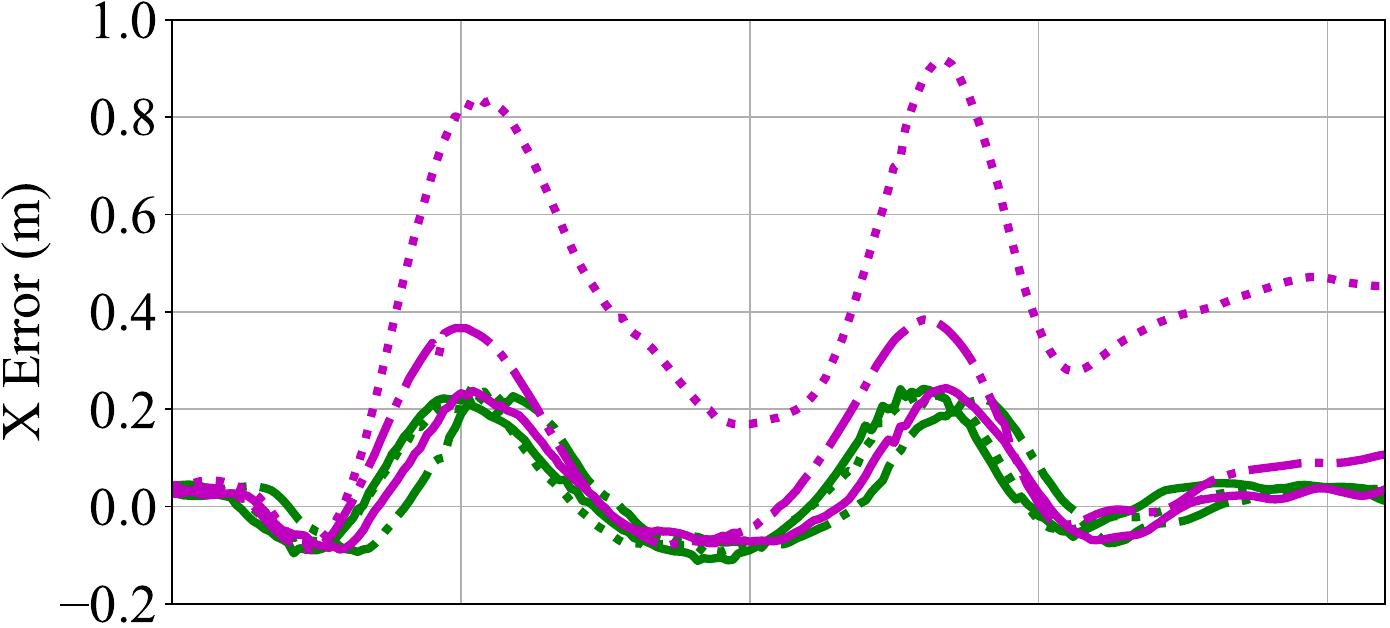}
    \end{minipage}\\[1ex]  
    \begin{minipage}{\linewidth}
        \includegraphics[width=\linewidth]{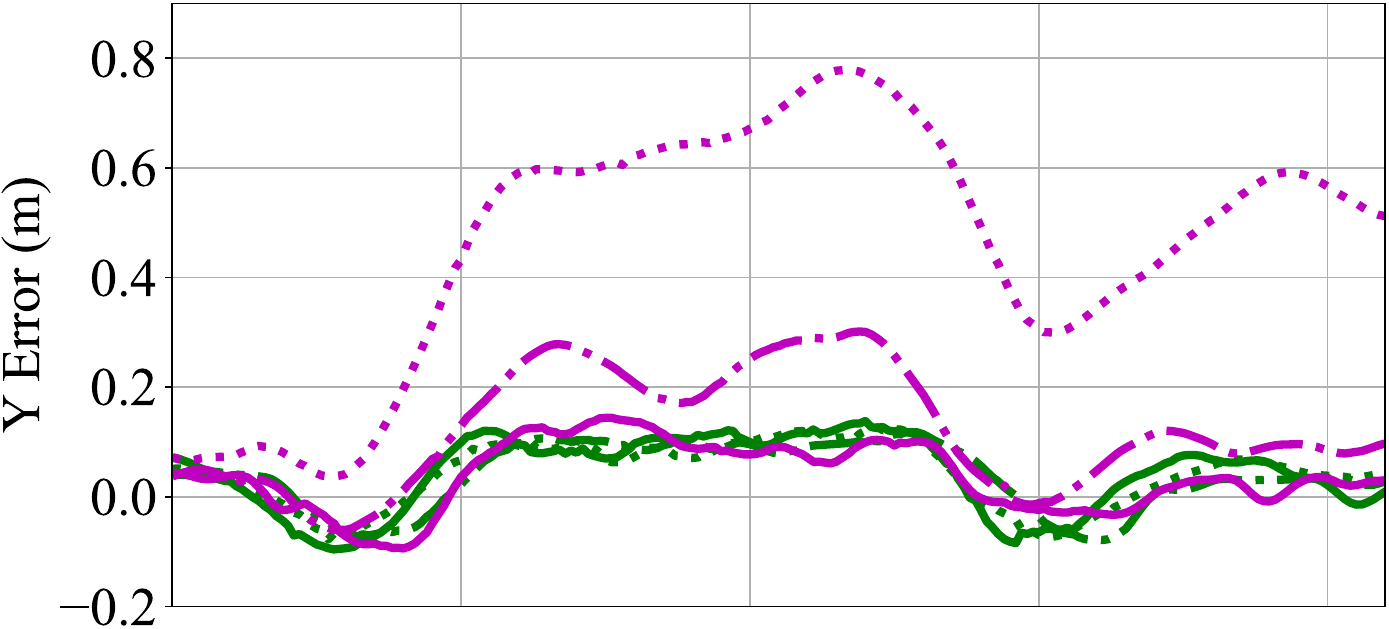}
    \end{minipage}\\[1ex]
    \begin{minipage}{\linewidth}
        \includegraphics[width=\linewidth]{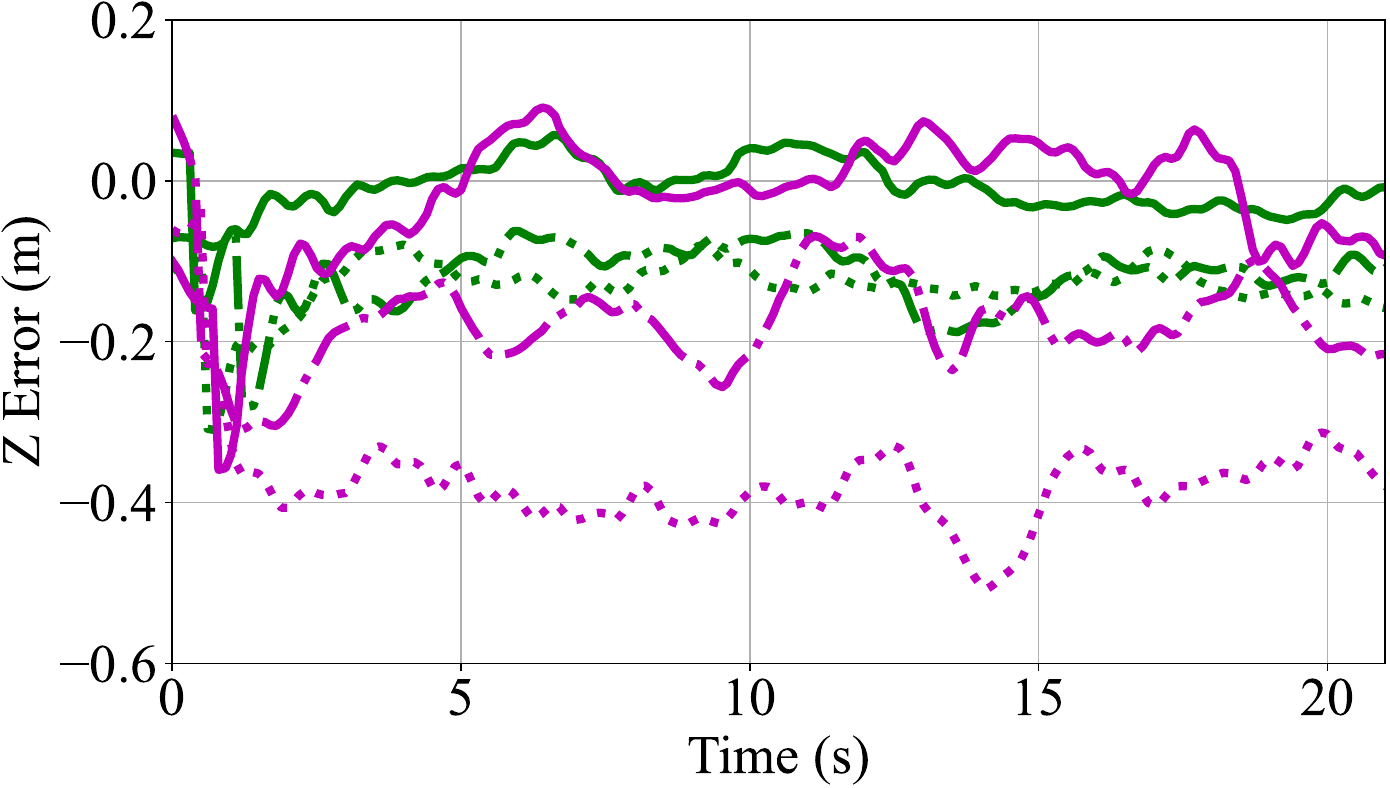}
    \end{minipage}
    \begin{minipage}{\linewidth}
    \centering

\includegraphics[width=0.75\linewidth]{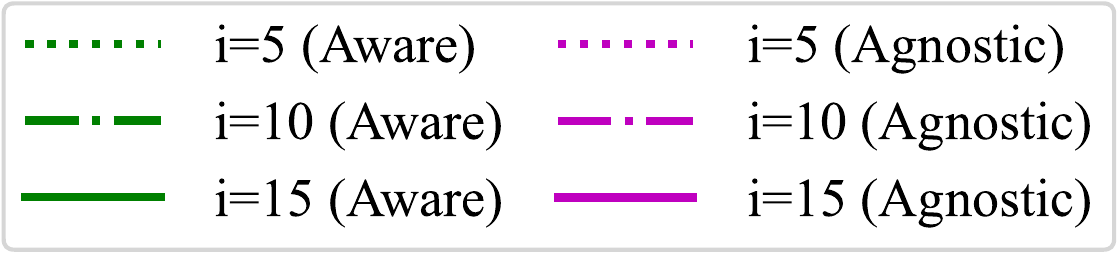}
    \end{minipage}    \caption{Comparsion of the tracking error for computation-aware and computation-agnostic controllers.}\label{fig.tracking error}
\end{figure}

As shown in Section \ref{sec.stability analysis}, $\dot{V}(s)$ can be bounded by the sum of three terms $\dot{\bar{V}}_i$, ${\dot{V}^{\textup{comp}}}_i$ and ${\dot{V}^{\textup{math}}}$, which are defined as 
\begin{equation}\nonumber
\begin{aligned}
\dot{\bar{V}}_i(s) &= s^{\top} \left( 
m{g_v}+R{{f}}_u + \hat{f}_d - m\ddot{p}_d + \Lambda m \dot{\Tilde{p}} 
\right),\\ {\dot{V}^{\textup{math}}}(s)&= |s|^{\top} \cdot B_f \cdot \sigma^{\textup
{math}} 
,\\
{\dot{V}^{\textup{comp}}}_i(s) &= |s|^{\top} \cdot B_f \cdot \sigma_i^{\textup
{comp}},
\end{aligned}
\end{equation}
where \(\hat{f}_d\), \(\sigma^{\textup{math}}\), and \(\sigma_i^{\textup{comp}}\) denote the posterior mean and the mathematical and computational uncertainties of the learned GP model of $f_d$, respectively. Additionally, $B_f$ is defined as the upper bound of the RKHS norm according to Lemma~\ref{lemma.worst case error}. To avoid solving the SOCP problem, we utilize Proposition \ref{prop.explicit form} to construct an explicit control policy. For comparison purposes, we also synthesize another explicit form by setting $\sigma_i^{\textup{comp}}$ to zero, referring to the resulting controller as a computation-agnostic controller.

In Fig.~\ref{fig.tracking error}, we plot the tracking error for both the computation-aware and computation-agnostic controllers across different CG iterations. As shown in Fig.~\ref{fig.tracking error}, as the number of CG iterations decreases, the tracking error for the computation-agnostic controller increases significantly, indicating that constrained computation can indeed negatively impact tracking performance. However, the degradation is slight for our proposed computation-aware controller, and even 5 CG iterations achieve better performance in the $X$, $Y$, and $Z$ directions than the 10 CG iteration case for the computation-agnostic controller.

Interestingly, for the computation-aware controller, the position error in the $X$ and $Y$ directions with 5 and 10 CG iterations can even be better than with 15 CG iterations. This may be due to two reasons: First, slosh dynamics primarily cause effects akin to changes in the quadrotor's gravity, thus greatly affecting the $Z$ direction; Second, as we have mentioned, our attitude dynamics employ a highly robust controller that is distinct from the position dynamics, potentially compensating for inaccuracies in the position dynamics control loop. However, it is important to note that for the computation-agnostic scenario, due to the lack of consideration for computational uncertainty, the $X$ and $Y$ errors still increase greatly as the number of CG iterations decreases.

\begin{table}[t] 

\renewcommand{\arraystretch}{1.5}

\caption{{Root Mean Square Error of Position Tracking}}

		\centering

{

{	\begin{tabular}{c c c c c c c}

		\hline
            
		\hline

		 & \multicolumn{3}{c}{Computation-Aware} & \multicolumn{3}{c}{ Computation-Agnostic} \\

    \hline

		  & $i=5$ & $i=10$  & $i=15$ & $i=5$ & $i=10$  & $i=15$
    \\

		 \hline

		$\tilde{p}_x$ (m) & 0.0915 & 0.0911  & 0.0977 &  0.467 & 0.161  & 0.0949
          \\

		\hline 

		$\tilde{p}_y$ (m) & 0.0664 & 0.0655  & 0.0752 & 0.509 & 0.154  & 0.0685 \\

     \hline 

	    $\tilde{p}_z$ (m) & 0.142 &  0.122 & 0.0387 & 0.387  & 0.197  & 0.0765 \\

		 \hline

\end{tabular}}}
\label{tab.tracking error}
\end{table}

The average computation time for each control loop for GP dynamical model learning is presented in TABLE \ref{tab.computation time}. As shown in TABLE \ref{tab.computation time}, the computation time for dynamical GP learning increases with the number of CG iterations. For a fixed $i$, the computation-aware GP exhibits a marginal increase in computation time as it additionally needs to quantify computational uncertainty, compared to the computation-agnostic counterpart. However, given its improvement in tracking performance (with a reduction in tracking error on the $X$, $Y$, and $Z$ axes by approximately $0.37m$, $0.44m$, and $0.24m$ respectively for
$i=5$), the trade-off is worthwhile.

\begin{table}[t] 

\renewcommand{\arraystretch}{1.5}

\caption{{Average Computation Time for Each Control Loop}}

\label{table:error_p_2}

		\centering

{

{	\begin{tabular}{c c c c c c c}

		\hline
            
		\hline

		 & \multicolumn{3}{c}{Computation-Aware} & \multicolumn{3}{c}{ Computation-Agnostic} \\

    \hline

		  & $i=5$ & $i=10$  & $i=15$ & $i=5$ & $i=10$  & $i=15$
          \\

		\hline 
		GP (s)& 0.0155 & 0.0327  & 0.0452 & 0.0152 & 0.0311  & 0.0434 
  \\
     \hline 

\end{tabular}}}
\label{tab.computation time}
\end{table}

\section{Conclusions and Discussions} 
\label{sec.conclusion}
\textbf{Conclusion: }
In this paper, we propose a computation-aware framework for GP dynamical model learning to ensure stable control in robotic systems subjected to computational constraints. We thoroughly investigate the impacts of constrained computations on model learning errors by utilizing Gaussian processes. We find that computational errors in model learning are inevitable and can lead to deterioration in control performance. Subsequently, we quantify the consequences of these computational errors on system stability by evaluating the region of attraction. Finally, we present a novel, robust controller design methodology that incorporates these computational considerations using second-order cone programming. The effectiveness of the framework is demonstrated through several canonical control tasks, elucidating its potential to enhance control performance under computational constraints.

\textbf{Discussion: } 
It's crucial to emphasize that our method extends beyond stable control in robotics. Issues such as safety and passivity in robotic systems, which are characterized by energy-like functions, can also be addressed in a similar manner. A more detailed explanation can be found in Appendix \ref{appendix.energy shaping control}. Besides, the primary idea behind integrating computational uncertainty into GP is the probabilistic approximation of the inverse of the Gram matrix \cite{hennig2022probabilistic,wenger2022posterior}. Although this idea is deeply rooted in the architecture of GPs, the innovative concept of computational uncertainty may offer potential for extension to other probabilistic models, 
including Bayesian neural networks \cite{hernandez2015probabilistic} and deep Gaussian processes \cite{damianou2013deep}.

\section*{Acknowledgements}
This work was supported by the Consolidator Grant ``Safe
data-driven control for human-centric systems" (CO-MAN) of
the European Research Council (ERC) under grant agreement IDs 864686.
We would also like to express our gratitude to Ribhav Ojha and Samuel Belkadi, undergraduate students at the University of Manchester, for their invaluable assistance in building and assembling the experimental hardware.

\bibliographystyle{IEEEtran}
\bibliography{main}

\appendix
\subsection{Proof of Lemma \ref{lemma.worst case error}}\label{appendix.proof of lemma 1}

\begin{proof}
First, we will prove that the latent function \(f + gu \in
\mathcal{H}_{k}\), where $\mathcal{H}_{k}$ is the RKHS associated with GP kernel $k$. Under Assumption \ref{assump.bounded RKHS}, we have $f(x)=\sum_{p=1}^{\infty} \alpha_p k^f\left(x, x_p\right)$, $g(x)=\sum_{q=1} ^{\infty} \beta_q k^g\left(x, x_q\right)$, for some real coefficients \( \{\alpha_p\},  \{\beta_q\}\) and points \( \{x_p\}, \{x_q\} \). Thus, the function \( f(x)+g(x)u \) can be expressed as
\begin{equation}\nonumber
f(x) + g(x)u = \sum_{p=1}^{\infty} \alpha_p k^f(x, x_p) + u \sum_{q=1}^{\infty} \beta_q k^g(x, x_q).
\end{equation}
Because the GP kernel function corresponds to a composite kernel $k(x, u, x^{\prime}, u^{\prime}) = k^f(x, x^{\prime}) + uk^g(x, x^{\prime})u^{\prime}$, we have
\begin{equation}\nonumber
\begin{aligned}
k^f(x, x^{\prime}) &= k(x, u, x^{\prime}, 0),
\\
uk^g(x, x^{\prime}) &= k(x, u, x^{\prime}, 1) - k(x, u, x^{\prime}, 0).
\end{aligned} 
\end{equation}
To show that \( f+gu \in \mathcal{H}_k \), we need to express \( f+gu \) as a linear combination of \( k \). Setting \( u^{\prime} = 1 \) in the definition of \( k \), we get
\begin{equation} \nonumber
\begin{aligned}
&f(x) + g(x)u \\
=& \sum_{p = 1}^{\infty} \alpha_p k(x, u, x_p, 0) + \sum_{q=1}^{\infty} \beta_q k(x, u, x_q, 1) 
\\
-& \sum_{q=1}^{\infty} \beta_q k(x, u, x_q, 0),    
\end{aligned}
\end{equation}
implying that \( f(x) + g(x)u \in \mathcal{H}_k \). Thus, there exists $B_{f,g} > 0$, such that $\|f + gu\|_{\mathcal{H}_k} \leq B_{f,g}$.
Then \eqref{eq.pointwise convergence} can be proved by leveraging Corollary 1 in \cite{wenger2022posterior}. 
\end{proof}

\subsection{Proof of Corollary \ref{corollary.discretization}}\label{appendix.proof of corollary 1}

\begin{proof}
By Assumption \ref{assump.bounded RKHS}, the function $f$ and $g$ are both Lipschitz continuous with Lipschitz constants $L_f$ and $L_g$ satisfying $L_f^2=2 B_f^2 \|k^f\|_{\infty}\left\|\frac{\partial k^f}{\partial x}\right\|_{\infty}$ and $L_g^2=2 B_g^2 \|k^g\|_{\infty}\left\|\frac{\partial k^g}{\partial x}\right\|_{\infty}$ ,
and bounded by
$\left\|f \right\|_{\infty} \leq B_f^2\|k^f\|_{\infty}$ and $\left\|g \right\|_{\infty} \leq B_g^2\|k^g\|_{\infty}$ \cite{berkenkamp2016safe,steinwart2008support}. 
For $x, x^{\prime} \in \mathcal{X}$, the derivative of the Lyapunov function $\dot{V}(x)$ satisfies
\begin{equation}\nonumber
\begin{aligned}
&\left|\dot{V}(x) - \dot{V}(x^{\prime}) \right|
\\
=& \left|\frac{\partial V(x)}{\partial x} \left( f(x)+g(x)\pi(x) \right) - \frac{\partial V(x^{\prime})}{\partial x^{\prime}} \left( f(x^{\prime})+g(x^{\prime})\pi(x^{\prime}) \right)
\right|
\\
\leq & \left|f(x)+g(x)\pi(x)\right| \cdot \left|\frac{\partial V(x)}{\partial x} - \frac{\partial V(x^{\prime})}{\partial x^{\prime}} \right|
\\
&+\left|\frac{\partial V(x^{\prime})}{\partial x^{\prime}} \right| \cdot \left|f(x)+g(x)\pi(x) - f(x^{\prime}) - g(x^{\prime})\pi(x^{\prime})\right|
\\
\leq& \left(B_f^2 \|k^f\|_{\infty} + B_g^2 B_{\pi} \|k^g\|_{\infty}  \right) \cdot \left\|\frac{\partial^2 V(x)}{\partial x^2}\right\|_{\infty}\cdot |x - x^{\prime}| 
\\
&+ \sqrt{2}\left\|\frac{\partial V(x)}{\partial x}\right\|_{\infty}\cdot \left(  B_f \sqrt{ \|k^f\|_{\infty}\left\|\frac{\partial k^f}{\partial x}\right\|_{\infty}} \right) \cdot|x - x^{\prime}|
\\
&+ \sqrt{2}\left\|\frac{\partial V(x)}{\partial x}\right\|_{\infty}\cdot \left(  B_g \pi_g  \sqrt{ \|k^g \|_{\infty}\left\|\frac{\partial k^g}{\partial x}\right\|_{\infty}} \right) \cdot|x - x^{\prime}|
\\
\triangleq& L \cdot |x - x^{\prime}|,
\end{aligned}    
\end{equation}
where $L$ is the Lipschitz constant of $\dot{V}$. Then, \eqref{eq.sufficient condition discretization} is the sufficient condition of \eqref{eq.sufficient condition}, and this corollary can be proved by using Theorem \ref{theorem.ROA}.
\end{proof}

\subsection{Energy Shaping Control}\label{appendix.energy shaping control}
In this section, we demonstrate that the computation-aware learning framework is not limited to stable control tasks. Indeed, it can be extended to analyze and achieve safety and passivity based on energy-like functions, which we refer to as energy-shaping control.

Energy shaping control is extensively used in controller design. Its core idea is to synthesize a controller based on an energy-like function, reshaping the system's natural energy to achieve desired control performance. For example, in regulation and tracking tasks, Lyapunov functions are widely used to design stable controllers \cite{wen1990unified,santibanez1997strict,khansari2014learning,giesl2015review,nguyen2015optimal}. In reach-avoid tasks, barrier functions or potential functions are commonly employed to ensure the safety of controllers \cite{hsu2015control,rauscher2016constrained,ames2019control,dawson2022safe,xiao2023barriernet}. Additionally, in the control of Euler-Lagrange models, storage functions are used to achieve passivity \cite{zhao2008passivity,hatanaka2015passivity,capelli2022passivity}. 

The performance of a given controller, denoted as $\pi(x)$, can be effectively evaluated by determining the invariant set within which key properties like stability, safety, and passivity are maintained.
\begin{definition}[Invariant Set]
A set $\mathcal{S}$ is said to be an invariant set if for all $x(0) \in \mathcal{S}$, we have $x(t) \in \mathcal{S},\, \forall t >0$. 
\end{definition}

For stability analysis, ROA is the invariant set that represents the area where a given control law \(\pi(x)\) guarantees the asymptotic stability of an equilibrium point \cite{henrion2013convex, topcu2009robust}. Besides, for safety analysis, the safe set \cite{ames2019control,dawson2023safe} is the invariant set that defines the states in which the system's operation remains within safe boundaries. This ensures that safety constraints are not violated, which is significant for safety-critical systems. The invariant set is typically related to a well-defined energy function:

\begin{definition}
[Energy Function and Power Function] 
\label{def.energy and power function}
For a given policy $\pi(x)$, an energy function $E(x) \in \mathbb{C}^2$ and a power function $P(x) \in \mathbb{C}^1$ are well-defined for set $\mathcal{S}$ if the following holds: for all $x \in \mathcal{S} = \left\{x \in \mathcal{X} | E(x) \in \mathcal{E} \right\}$, if the following condition is satisfied:
\begin{equation}\label{eq.energy variation}
\dot{E}(x)=\frac{\partial E(x)}{\partial x}\left( f(x)+g(x)\pi(x) \right) \leq P(x).
\end{equation}
Then, $\mathcal{S}$ is an invariant set.
\end{definition}

In Definition \ref{def.energy and power function}, 
we call \eqref{eq.energy variation} as the energy variation condition because 
it requires the time variation of the energy function $\dot{E}(x)$ to be slower than a power function $P(x)$. 

In control theory, the terms ``energy function'' and ``power function'' are not formally defined, but the concept of using energy-like functions for analyzing control systems is widespread. For stability analysis, the energy function \(E(x)\) is often chosen as the Lyapunov function, and the power function \(P(x)\) is set by \(P(x) = -\gamma(E(x))\), where \(\gamma\) is a class \(\mathcal{K}\) function. The invariant set \(\mathcal{S}\) is defined as the ROA of the system, given by \(\mathcal{S} = \{x \in \mathcal{X} \,|\, 0 \leq E(x) \leq c\}\), with \(c > 0\). For safety analysis, the negation of the barrier function is used as the energy function, and the power function is defined as \(P(x) = \alpha(E(x))\), where \(\alpha\) is an extended class \(\mathcal{K}_{\infty}\) function. The invariant set \(\mathcal{S}\) is considered as the safe set, denoted by \(\mathcal{S} = \{x \in \mathcal{X} \,|\, E(x) \leq 0\}\).    

Specifically, for mechanical systems represented by the input-output gravity-compensated Euler-Lagrange model \cite{ortega1998euler}:
\begin{equation}\nonumber
\begin{cases}
M(\chi) \ddot{\chi} + C(\chi, \dot{\chi}) \dot{\chi} + D \dot{\chi} + \frac{\partial P(\chi)}{\partial \chi} = u, 
\\
y = \dot{\chi},
\end{cases}  
\end{equation}
where \(\chi\) represents the system's position, and \(x = \begin{bmatrix} \chi, \, \dot{\chi} \end{bmatrix}\) is the system state, \(u\) and \(y\) are the control input and output respectively. For passivity analysis \cite{romeo1998passivity,capelli2022passivity}, the energy function is often selected as the storage function of the system, satisfying \(E(x) = P(\chi) + \frac{1}{2} \dot{\chi}^T M(\chi) \dot{\chi}\), and the power function is defined as \(P(x) = \pi(x) y\) for a given policy \(\pi(x)\). The power dissipated by the system, \(P_d(x) = \pi(x) y - \dot{E}(x)\), also considered as a passivity margin, should satisfy \(P_d(x) \geq 0\) for a passive system. In this case, the invariant set $\mathcal{S}$ is defined as \(\mathcal{S} = \left\{x\in \mathcal{X} | E(x) \geq 0 \right\}\).

Using the concept of energy functions and invariant sets, the stability analysis and controller design techniques presented in this paper can be easily extended to various energy shaping controls, including those for safety or passivity properties.



\end{document}